\documentclass[journal,12pt,onecolumn,draftcls,doublespace]{IEEEtran}
\usepackage{graphics}
\usepackage{multicol}
\usepackage{lipsum}
\usepackage{color}
\usepackage[cmex10]{amsmath}
\usepackage{amsthm}
\usepackage{amssymb}
\usepackage{mathrsfs}
\usepackage{mathtools}
\usepackage{amsbsy}
\usepackage{xcolor}
\usepackage{mathrsfs}
\usepackage{setspace}
\usepackage{float}
\usepackage{graphicx}
\usepackage{psfrag}
\usepackage{epstopdf}
\usepackage{lettrine}
\usepackage{psfrag}
\usepackage{newclude}
\usepackage[normalem]{ulem}
\usepackage{latexsym}
\usepackage{algpseudocode}
\usepackage{algorithm,algpseudocode}
\usepackage{algorithmicx}
\usepackage{multirow}
\usepackage{amsmath}
\usepackage{wasysym}
\usepackage{mathrsfs}
\usepackage{amsmath,epsfig,cite,amsfonts,amssymb,psfrag,color}

\ifCLASSOPTIONcompsoc
\usepackage[caption=false,font=normalsize,labelfon
t=sf,textfont=sf]{subfig}
\else
\usepackage[caption=false,font=footnotesize]{subfig}
\fi

\allowdisplaybreaks
\newtheorem{remark}{Remark}
\newtheorem{theorem}{Theorem}
\newtheorem{lemma}{Lemma}

\graphicspath{{figures/}}
\allowdisplaybreaks

\ifCLASSOPTIONcaptionsoff
  \newpage
\fi
\begin{document}
\title{SWIPT-Enabled Multiple Access Channel: Effects of Decoding Cost and Non-linear EH Model}
    \author{\IEEEauthorblockN{Pouria Nezhadmohammad, Mohsen Abedi, Mohammad Javad Emadi, and Risto Wichman}
    
\thanks{This research work was partly supported by Academy of Finland under the grant 334000.\par
P. Nezhadmohammad and M. J. Emadi are with the Department of Electrical Engineering, Amirkabir University of Technology, Tehran 159163-4311, Iran
(E-mails:\{pourianzh, mj.emadi\}@aut.ac.ir).\par
M. Abedi and R. Wichman are with the Department of Electrical Engineering, Aalto University, Espoo 02150, Finland (E-mails:\{mohsen.abedi, risto.wichman\}@aalto.fi).
}

}

	\maketitle

\maketitle
\begin{abstract}We studied power splitting-based simultaneous wireless information and power transfer (PS-SWIPT) in multiple access channels (MAC), considering the decoding cost and non-linear energy harvesting (EH) constraints at the receiving nodes to study practical limitations of an EH communication system. Under these restrictions, we formulated and analyzed the achievable rate and maximum departure regions in two well-studied scenarios, i.e., a classical PS-SWIPT MAC and a PS-SWIPT MAC with user cooperation. In the classical PS-SWIPT MAC setting, closed-form expressions for the optimal values of the PS factors are derived for two fundamental decoding schemes: simultaneous decoding and successive interference cancellation. In the PS-SWIPT MAC with user cooperation, the joint optimal power allocation for users as well as the optimal  PS factor are derived. This reveals that one decoding scheme outperforms the other in the classical PS-SWIPT MAC, depending on the function type of the decoding cost. Finally, it is shown that the cooperation between users can potentially boost the performance of a PS-SWIPT MAC under decoding cost and non-linear EH constraints. Moreover, effects of the decoding cost functions, non-linear EH model and channel quality between the users are studied, and performance characteristics of the system are discussed.
\end{abstract}
\IEEEpeerreviewmaketitle
\section{Introduction} 

\lettrine{E}{NERGY} harvesting (EH) is a promising technique in the next generation of wireless communication. The benefit it offers by energizing low-power devices using environmental energy resources makes it an emerging alternative technology in green wireless networking. Energy resources are classified into two types; ambient and dedicated \cite{ku2015advances}. The use of ambient resources, e.g., wind, solar, heat, etc. in the conventional EH systems has introduced a vast uncertainty regarding the harvested energy level, since most ambient resources are usable only in specific environments under special conditions. Therefore, to guarantee the quality of service  (QoS) and EH stability, dedicated resources such as radio frequency (RF) signals have emerged \cite{lu2014wireless}.\par

Recent studies in wireless networks demonstrate the ever-increasing demand for high energy efficiency in modern communication systems, especially in fifth-generation (5G) networks and beyond. Resultantly, a growing number of studies are being undertaken in search of efficient ways to remotely charge wireless devices. To this end, RF signals have exhibited a strong potential to not only transmit information but also to simultaneously transfer energy, underpinning the theory of simultaneous wireless information and power transfer (SWIPT) \cite{varshney2008transporting}. 
Specifically, since high data rates exponentially increase the power consumption of wireless devices and degrades the battery lifetime, SWIPT has become an effective technology to tackle the contradiction between the high data rate and energy conservation \cite{liang2019simultaneous}. In this direction, the performance of SWIPT is investigated in the new 5G frequencies \cite{zhai2018simultaneous} and multi-carrier non-orthogonal multiple access (MC-NOMA) networks \cite{tang2020decoupling}. Furthermore, authors in \cite{khalfet2019ultra} characterize information-energy capacity region in SWIPT binary symmetric channels. However, although the SWIPT systems offer unprecedented benefits to modern communication networks, they  suffer from fundamental performance limitations as well. To address these limitations, authors in \cite{khalfet2018simultaneous} analyzed the channel output feedback in a SWIPT-based two-user Gaussian interference channel. Unfortunately, the feedback at most doubles the energy transmission rate at high signal-to-interference-plus-noise ratio (SINR) in a SWIPT-based two-user Gaussian MAC \cite{amor2017feedback}.
 \par

To leverage the SWIPT systems in practice, four major techniques are proposed in the literature: time switching (TS), power splitting (PS), antenna switching (AS), and spatial switching (SS) \cite{krikidis2014simultaneous}. By employing TS, the receiver switches in time between information decoding (ID) and EH mode. In PS, 
    the receiver splits the received signal into two different streams with different power levels, one for EH and the other for ID (see \cite{krikidis2014simultaneous,ding2015application}), whereas the other two schemes are applied to multi-antenna configurations \cite{ding2015application}. 
In addition to these techniques, for the case of multi-user scenarios, authors in  \cite{amor2015simultaneous,amor2017feedback,khalfet2018simultaneous} propose a signal design method performed at the transmitter sides to not only transfer information to the receiver, but also to construct a specific signal for energy harvester, that means a signal splitting is applied at the transmitter. However, in the conventional PS-SWIPT concept, especially the single user communications or the broadcast scenarios, a specific signal design for energy transfer is not considered \cite{krikidis2014simultaneous,mishra2018energy,ma2020power,ding2015application}, i.e., the transmitted signal only conveys the information. Therefore, the destination performs the PS scheme to decode the message while harvesting energy as well.

To compare PS with other practical techniques especially TS, several scenarios are investigated in the literature. Authors in \cite{shen2014wireless} and \cite{kudathanthirige2019max} studied both PS and TS techniques for SWIPT-based multiple-input-single-output (MISO) and massive multiple-input-multiple-output (MIMO) channels, respectively. Furthermore, in a SWIPT-based decode-and-forward (DF) relay channel, \cite{liu2017joint}  investigates PS and AS techniques and  \cite{oshaghi2019throughput} takes into account the temperature constraints. 
Moreover, the problem of power allocation in SWIPT systems is studied with a cooperative NOMA \cite{yang2017impact}, PS NOMA \cite{tang2019joint}, and a two-user MAC \cite{yu2017energy}, whereas  \cite{huang2020power} and \cite{yang2020resource} analyze power allocation in a PS-SWIPT device-to-device underlaid cellular network  carried out from practical and theoretical points of view, respectively. 
 Interestingly, although several studies indicate that PS entails a higher complexity of implementation due to the PS factor optimization, it outperforms TS owing to the logarithmic form of the rate-energy function, especially in high signal-to-noise ratio (SNR) communications and delay constrained cases as the both energy and information are partly wasted in TS \cite{krikidis2014simultaneous,tang2020decoupling}.  
As a result, we concentrated on PS and disregarded  the TS scheme in this paper, since our goal is to enlarge the achievable rate/departure region.
 \par
 
Although a linear EH model has been conventionally considered in most literature, recent investigations exposed a resource allocation mismatch in the practical implementation approaches \cite{shi2017power}.
 To be realistic, EH hardware reveals a non-linear behavior due to physical limitations, such as storage imperfections, non-linear circuits, processing costs, non-ideal energy conversion efficiency, etc. \cite{xu2017simultaneous,varasteh2020capacity}. Accordingly, the realistic efficiency models of EH circuits must be considered which is investigated in \cite{alevizos2018sensitive}.
 Here, the practical SWIPT antennas with non-linear EH model are equipped with electronic devices to convert the received RF signal power to a direct current (DC) through impedance matching and passive filtering \cite{boshkovska2015practical,clerckx2018fundamentals}.
 To show that PS still has superiority over TS, authors in \cite{xiong2017rate} proved that in a SWIPT-based MIMO broadcasting channel with a non-linear EH model, PS achieves a larger average rate-region than TS. Similarly, authors in \cite{kim2019experimental} experimentally showed that at a high SNR level, energy-throughput regions of the PS receiver are considerably larger than the TS receiver. Additionally, authors in  \cite{jiang2019power} remark that the PS receiver architecture in hybrid access point (H-AP) network requires less power to meet the information and energy requirements at the user's end, and resultantly, more PS users can be served than the ones with TS. However, owing to a higher input RF power in TS, the EH efficiency of TS is larger than that of PS users. In practice, the EH operation is observed to be non-linear due to the saturation effect\cite{wang2020performance}. However, this effect may be rectified by placing several EH circuits in parallel, yielding a sufficiently large linear  conversion  region  in  practice;  thus, the EH-conversion  function  can  also  be  modeled with a linear function  \cite{kang2018joint,ma2019generic}. Authors in \cite{liu2016swipt} claim that the RF-EH circuits can only harvest energy when the received power is higher than a specific sensitivity level. As a result, determining an optimal PS factor in PS-SWIPT systems is instrumental in order to maximize the throughput and minimize the energy consumption.
  \par

On the other hand, the required power to decode the incoming data rate has practically a significant impact on the performance of the low power EH systems.
 Generally,  \emph{decoding cost}  at the destination depends on processing algorithms, circuit designs, and architectures which tends to dominate the other sinks of power (e.g. ADC, DAC, encoding, modulation/demodulation, amplification, etc.) for a high data rate communication \cite{marcu200990}. Authors in \cite{grover2011towards} indicated that in VLSI devices, the decoding cost power is lower-bounded by the number of message-passing iterations among processing elements $n$, amount of energy consumed at each iteration $E_{node}$, the number of computational nodes $m$, the channel dispersion, and the desired upper-bound on the decoding error probability $P_e$. 
The impact of  decoding cost on the performance of the conventional EH systems is extensively studied in the literature \cite{arafa2015optimal,abedi2018cooperative,arafa2016energy,arafa2017energy,abedi2018power,mahdavi2013energy}. Decoding cost as a new constraint was recently investigated in point-to-point, MAC, and broadcast channels with offline EH policies \cite{arafa2015optimal}. Furthermore, the decoding cost effect is considered in EH point-to-point channel with an energy helper \cite{abedi2018cooperative}, EH MAC \cite{arafa2015optimal}, EH two-way channels \cite{arafa2017energy}, data cooperation between users \cite{arafa2016energy}, and SWIPT-based two-hop half-duplex DF relay channel \cite{abedi2018power}. Authors in \cite{grover2011towards} proved that by applying VLSI and LDPC codes, the decoding cost appears as an increasing quadratic function of the decoding rate. However, unlike the EH conversion function, many different codes with multiple decoding algorithms for each code yields a variety of decoding cost function.
In all recent publications \cite{arafa2015optimal,abedi2018cooperative,arafa2016energy,arafa2017energy,abedi2018power},  the decoding cost is assumed as an increasing convex function of data rate to make the optimization problems convex and tractable \cite{mahdavi2013energy}.
However, a general non-decreasing decoding cost function in terms of the data rate has not yet been studied in recent works.

A review of the existing literature did not yield any work that investigates the impact of the decoding cost in the PS-SWIPT MAC scenarios under non-linear EH receivers. To capture the practical restrictions of the physical system in this paper, we focus on the problem of PS-SWIPT by taking into account the non-decreasing decoding costs incurred at the non-linear EH nodes. The predominant potential contributions of this paper are twofold. First, we maximize the rate departure region for a two-user classical PS-SWIPT MAC with two different decoding strategies, i.e., simultaneous decoding and successive interference cancellation (SIC). Subsequently, we formulated the maximum departure region boundaries (MDRBs) with power allocation in a PS-SWIPT MAC with user cooperation. Second, we derived closed-form expressions for optimal PS factors and sum rates in the two PS-SWIPT MAC scenarios. Unlike the previous works wherein an increasing convex form of decoding cost is assumed to ease the optimization problems,  we assume a general non-decreasing form of decoding cost function. Resultantly, it is proven that although our problems are non-convex, global solutions can be derived in closed-form  via solving nonlinear system of equations. Finally, the simulation results are presented, which provide invaluable insights into the PS-SWIPT systems with regard to the decoding costs and non-linear EH model.\par
The rest of the paper is structured as follows. Section II
describes Classical PS-SWIPT MAC. Section III presents the PS-SWIPT MAC with user cooperation. Finally, the numerical results and conclusions are included in Sections IV and V, respectively.

\section{Classical PS-SWIPT MAC}
\subsection{System Model}
 We study a two-user fading MAC wherein the destination employs PS technique to harvest energy from the received signals and utilize the harvested energy for decoding the users’ information. As depicted in Fig. 1, users 1 and 2 encode the messages $m_1$ and $m_{2}$, respectively, then the users' signals are transmitted to the destination. To enlarge the achievable rate regions for both of the classical PS-SWIPT MAC and the cooperative one, it is assumed that both users simultaneously transmit their signals. Thus, in what follows, all the achievable rate regions are derived based on the NOMA transmissions. Besides, perfect  channel state information (CSI) is assumed at the nodes to derive performance upper bound of the considered systems. Here, the destination aims to efficiently harvest energy from the received signal in the EH block for the decoding process to maximize the achievable rate region. 

The received signal at the destination is
{{\begin{equation}\label{MacsimY}
		\begin{aligned}
		& Y =  h_1 X_1+h_2 X_2 + Z,
		\end{aligned}
\end{equation} }}where $h_1$ and $h_2$ indicate the channel gains, $X_1$ and $X_2$ are the transmitted symbols \footnote{Since the focus of our work is to consider the decoding costs effects on the performance of the SWIPT MAC with and without cooperating users in presence of non-linear EH destination, we have only considered the PS scheme at the destination without  considering cooperative signal design for the energy transfer \cite{khalfet2018simultaneous,amor2017feedback}. Our approach could be extended to consider coherent energy-signal transmission by the two users, as well.} by the first and the second user with $\mathbb{E}\lbrace{\vert{X_1}\vert}^{2}\rbrace=P_1$ and $\mathbb{E}\lbrace{\vert{X_2}\vert}^{2}\rbrace=P_2$, and $Z \sim \mathcal{CN}(0,N)$ is complex additive white gaussian noise (AWGN) at the destination.
The destination splits the received signal into EH stream $Y_{EH}$ and ID one $Y_{ID}$ as
{{\begin{IEEEeqnarray}{rlu}
			& Y_{EH}=\sqrt{\rho}Y, \hspace{0.1cm}Y_{ID}=\sqrt{1-\rho}Y + Z_p,
\end{IEEEeqnarray} }}where $\rho\in [0,1]$ refers to the unit-less PS factor and $Z_p\sim \mathcal{CN}(0,N_p)$ represents the signal processing noise at the power splitter. So, the harvested power at the destination is given by
{{\begin{IEEEeqnarray}{rlu}
			& P_{EH}=\psi\big(\rho(|h_1|^2 P_1+|h_2|^2 P_2 + N)\big),
\end{IEEEeqnarray}}}wherein $\psi(.)$  models the output DC power of the rectifier in terms of the input RF power $P_{in}$ defined by 
{{\begin{IEEEeqnarray}{rlu}
			& \psi(P_{in})=\psi^{DC}(P_{in})= \dfrac{\Psi(P_{in})-P_{max}^{DC}\theta}{1-\theta}, 
\end{IEEEeqnarray} }}which is a zero-crossing, non-decreasing \cite{mishra2018energy}, and upper-bounded EH conversion function with the logistic function 
{{\begin{IEEEeqnarray}{rlu}\label{rectifier:q1q2}
			& \Psi(P_{in})= \dfrac{P_{max}^{DC}}{1+e^{-q_1(P_{in}-q_2)}}, 
\end{IEEEeqnarray} }}where $P_{max}^{DC}$ refers to the maximum output DC power and $\theta=(1+e^{q_1q_2})^{-1}$ is a parameter to ensure zero input/zero output of the EH block. Practically, the constant values   $q_1$ and $q_2$ depend on the  hardware elements of the rectifier and are determined via curve fitting methods, \cite{xiong2017rate,boshkovska2015practical}.

\begin{remark}\label{remark:psi_linearity}
In case of placing several EH circuits in parallel, the  EH conversion function is modeled with a linear zero-crossing function $\psi(P_{in})=\eta P_{in}$, wherein $\eta\in [0,1]$ refers to the fixed energy conversion efficiency, \cite{kang2018joint,ma2019generic}.
 
\end{remark}

On the other hand, the required decoding power in the ID block is modeled as $\varphi(R)$ which is a non-decreasing function of the achievable rate $R$.
\begin{figure}[!t]
\centering
	\includegraphics[scale=0.65]{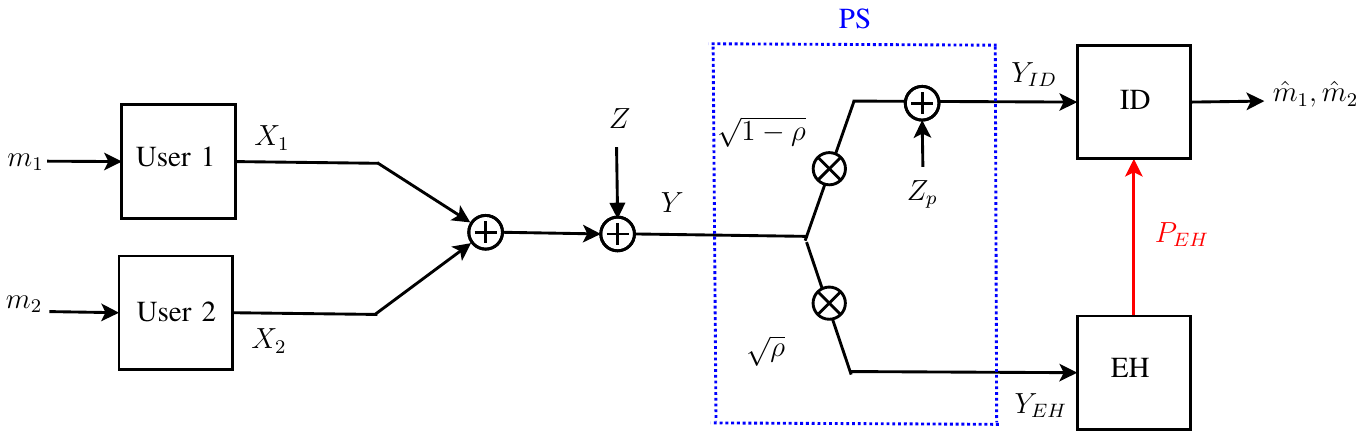}
\caption{ Classical PS-SWIPT MAC.}\label{sysmod2}
\end{figure}

\subsection{Simultaneous Decoding Scheme}
Assume that the destination applies the well-known simultaneous decoding strategy to recover the users' messages \cite{el2011network}. Therefore, the achievable rate region is constrained by
{{\begin{equation}\label{simulr1}    
		\begin{aligned}
		\hspace{-1.8cm} R_1 \leq \frac{1}{2}\log\left(1+\dfrac{(1-\rho)|h_1|^2P_1}{(1-\rho)N+N_p}\right),
		\end{aligned}
		\end{equation}
		\begin{equation}\label{simulr2}
		\begin{aligned}
		\hspace{-1.8cm} R_2 \leq  \frac{1}{2}\log\left(1+\dfrac{(1-\rho)|h_2|^2P_2}{(1-\rho)N+N_p}\right),
		\end{aligned}
		\end{equation}
		\begin{equation}\label{simuldes}
		\begin{aligned}
		\hspace{-1cm} R_1+R_2 \leq \frac{1}{2}\log\left(1+\dfrac{(1-\rho)(|h_1|^2P_1+|h_2|^2P_2)}{(1-\rho)N+N_p}\right),
		\end{aligned}
		\end{equation}
		\begin{equation}\label{simulphi}
		\begin{aligned}
		\hspace{-0.4cm} R_1+R_2 \leq \varphi^{-1}\Big(P_{EH}\Big)=&\varphi^{-1}\Big(\psi\big(\rho(|h_1|^2P_1+|h_2|^2 P_2\\
	    &+ N)\big)\Big),
		\end{aligned}
		\end{equation} }}where (\ref{simulr1}) and (\ref{simulr2}) refer to the achievable rate constraints for the first and the second user, respectively, while (\ref{simuldes}) denotes the sum rate constraint. Also, the decoding cost constraint (\ref{simulphi}) utilizes the fact that $\varphi(.)$ is a non-decreasing function.

\begin{theorem}\label{simultheorem}
 The MDRB is formed by
\begin{equation}
			\begin{aligned}\label{simulregion11}
\hspace{-0.1cm}\begin{cases}

			 R_1^*=\varphi^{-1}\big(\psi(\rho a)\big)-R_2^*, \\
			 		\hspace{5.5cm}	 \rho_1^*\leq\rho\leq \rho_c^*\\ 
	 R_2^*=\frac{1}{2}\log\left(1+\dfrac{(1-\rho)|h_2|^2P_2}{(1-\rho)N+N_p}\right),

\end{cases}
			\end{aligned}
\end{equation} 
\begin{equation}\label{simulregion2}
\begin{aligned}
\hspace{-0.1cm}\begin{cases}
			R_2^*=\varphi^{-1}\big(\psi(\rho a)\big)-R_1^*, \\
			\hspace{5.5cm} \rho_2^*\leq\rho\leq \rho_c^*\\
			  R_1^*=\frac{1}{2}\log\left(1+\dfrac{(1-\rho)|h_1|^2P_1}{(1-\rho)N+N_p}\right),
\end{cases}
\end{aligned}
\end{equation}
and
\begin{equation}\label{simulregion33}
\hspace{-2cm}R_1^*+R_2^*=\frac{1}{2}\log\left(1+\dfrac{(1-\rho_c^{*})(a-N)}{(1-\rho_c^{*})N+N_p}\right),
\end{equation}
where
{{\begin{equation}
			\begin{aligned}\nonumber
			\hspace{-0.9cm}&\rho_c^*=\Gamma_c^{-1}(N_p), \rho_1^*=\Gamma_1^{-1}(N_p),  \rho_2^*=\Gamma_2^{-1}(N_p),
			\end{aligned}
			\end{equation} }}
and
{{\begin{equation}
			\begin{aligned}\nonumber
			\hspace{-0.4cm}&\Gamma_c(x)=2^{2\varphi^{-1}\big(\psi (x a)\big)}\Big((1-x)N+N_p\Big)-(1-x)a,\\
			\hspace{-0.4cm}&\Gamma_1(x)=\Gamma_c(x)+(1-x)|h_1|^2P_1,\\
			\hspace{-0.4cm}&\Gamma_2(x)=\Gamma_c(x)+(1-x)|h_2|^2P_2,\\
			\hspace{-0.4cm}& \hspace{0.65cm}a =|h_1|^2P_1+|h_2|^2P_2+N.
			\end{aligned}
			\end{equation} }}
\end{theorem}
\begin{proof}
To form the MDRB constrained by (\ref{simulr1})-(\ref{simulphi}), we find the pairs of rates ($R_1$, $R_2$) who are infeasible to be added at the same time. Since $\varphi(.)$ and $\psi(.)$ are non-decreasing functions, the right terms of (\ref{simulr1}), (\ref{simulr2}), and (\ref{simuldes}) are decreasing, while the right term of (\ref{simulphi}) is increasing in $\rho$. The pair of rates  satisfying both (\ref{simuldes}) and (\ref{simulphi}) with strict inequality does not fall on the MDRB, as 
we can decrease $\rho$ with small enough $\epsilon>0$ so that all the constraints (\ref{simulr1})-(\ref{simulphi}) are satisfied with strict inequality. Then, we can feasibly increase both $R_1$ and $R_2$ with $\epsilon/M$ for a large enough $M>0$, which contradicts the definition of the MDRB. As a result, we consider three assumptions: 
 I) Assume (\ref{simuldes}) dominates (\ref{simulphi}), i.e., (\ref{simulphi}) is satisfied with strict inequality, while (\ref{simuldes}) is satisfied with equality. To show the contradiction here, we can feasibly  decrease $\rho$ with $\epsilon$ and increase both $R_1$ and $R_2$ with $\epsilon/M$;
 II) Assume (\ref{simulphi}) dominates (\ref{simuldes}).
 With this assumption, we have four cases. First, (\ref{simulr1}) and (\ref{simulr2}) are satisfied with equality. This case is already infeasible as it violates (\ref{simuldes}). Second, (\ref{simulr1}) and (\ref{simulr2}) are satisfied with strict inequality. Here, we can feasibly increase $\rho$, $R_1$, and $R_2$ with $\epsilon$, $\epsilon/M$, and  $\epsilon/M$, respectively, indicating  that this case also does not form the MDRB. Third, (\ref{simulr2}) dominates (\ref{simulr1}). This case, formulated in (\ref{simulregion11}), partly generates the MDRB as any small change in $\rho$ results in a violation of (\ref{simulr2}) or (\ref{simulphi}).
 Forth, (\ref{simulr1}) dominates (\ref{simulr2}). This case formulated in (\ref{simulregion2}), similarly generates another part of the MDRB;
III) Assume both (\ref{simuldes}) and (\ref{simulphi}) are satisfied with equality. This assumption formulated in (\ref{simulregion33}),
also partly creates the MDRB as any change in $\rho$ violates (\ref{simuldes}) or (\ref{simulphi}).
\end{proof}
Now, we characterize the convexity of the maximum departure region, bounded by the MDRB in Theorem \ref{simultheorem}.
\begin{remark} Since $\varphi(.)$ and $\psi(.)$ are both increasing, we derive  $d R_1^*/d R_2^*\leq -1$ on the MDRB part, implying that the maximum sum rate subject to (\ref{simulregion11}) occurs in $\rho_c^*$. Using similar arguments, the maximum sum rate subject to (\ref{simulregion2}) occurs also in $\rho_c^*$.
Moreover, if $\varphi^{-1}\big(\psi(x)\big)$ is a convex function within the interval $a\rho_1^*\leq x \leq a\rho_c^*$, then $d^2R_1^*/ d R_2^{*^2}\leq 0$ and the MDRB part calculated by (\ref{simulregion11}) is convex one. Otherwise, using time sharing technique, the convex hull of (\ref{simulregion11}) yields the maximum departure region. The same argument is verified for (\ref{simulregion2}).
\end{remark}
Depart from the maximum departure region, an important goal in PS-SWIPT MAC is to optimize the PS factor for  achievable sum rate  maximization. To this end, the optimization problem is formulated as
\begin{align}
\hspace{-0.3cm}\mathbf{P} \mathbf{1}: \max _{{\rho,R_1,R_2}} & R_1+R_2 \\
\text { {s.t.} }\hspace{0.2cm}&(\ref{simulr1})- (\ref{simulphi}), \hspace{0.1cm} \text{and} \hspace{0.1cm} 0\leq\rho\leq1, 
\end{align}
which is a non-convex optimization problem whose solution is a point on the MDRB calculated in Theorem \ref{simultheorem}.
\begin{lemma}\label{lm2} 
Optimally in $\mathbf{P} \mathbf{1}$, the whole harvested energy is consumed at the ID block, i.e., (\ref{simulphi}) is satisfied with equality. Moreover, (\ref{simuldes}) is also satisfied with equality.
\end{lemma}

\begin{proof}
Suppose (\ref{simulphi}) is satisfied with strict inequality. 
The right terms of (\ref{simulr1})-(\ref{simuldes}) and the right term of (\ref{simulphi}) are decreasing and increasing in $\rho$, respectively. So, we can reduce $\rho$ with a small enough $\epsilon>0$ such that all the constraints (\ref{simulr1})-(\ref{simulphi}) are satisfied with strict inequality. Thus, both $R_1$ and $R_2$ can be added with no constraints violation, which enhances the sum rate and contradicts the optimality.
On the other hand, suppose (\ref{simuldes}) is satisfied with strict inequality. Similarly, we can feasibly increase $\rho$ in a small enough value, letting $R_1$ and $R_2$ increase, and therefore the sum rate. Thus, Lemma \ref{lm2} follows. 	
\end{proof}

\begin{theorem}\label{simultaneous_sumrate_theorem}
The optimal PS factor $\tilde{\rho}$  in $\mathbf{P} \mathbf{1}$ is derived by solving the equation
\begin{equation}\label{rho*}
    \varphi \Big(\frac{1}{2}\log\big(1+\dfrac{(1-\rho)(a-N)}{(1-\rho)N+N_p}\big)\Big)=\psi\big(\rho a\big).
\end{equation}Then, the optimal sum-rate is given by
\begin{equation}\label{sum-rate-optimals}
   \tilde{R_1}+\tilde{R_2} =\frac{1}{2}\log\Big(1+\dfrac{(1-\tilde{\rho})(a-N)}{(1-\tilde{\rho})N+N_p}\Big),
\end{equation}
wherein $0 \leq \tilde{\rho}\leq 1$ and $\tilde{R_1}+\tilde{R_2}$ exist and are unique.
\end{theorem}

\begin{proof}
From Lemma \ref{lm2}, the right terms of (\ref{simuldes}) and (\ref{simulphi}) are equal in the optimal solution of $\mathbf{P} \mathbf{1}$, and thus, $\tilde{\rho}$ satisfies  (\ref{rho*}).  The left side of (\ref{rho*}) is decreasing  with  $\rho$ with strictly positive and zero values by setting $\rho=0$ and $\rho=1$, respectively. On the other hand, the right side of (\ref{rho*}) is increasing with  $\rho$ with zero and strictly positive values by setting $\rho=0$ and $\rho=1$, respectively. Hence, the optimal PS factor $0\leq \tilde{\rho}\leq 1$ exists and is unique. Then, we derive the 
the optimal sum rate in (\ref{sum-rate-optimals}) by substituting $\tilde{\rho}$ in the right term of (\ref{simuldes}).
\end{proof}

\begin{remark}
It is revealed that in simultaneous decoding scheme, $\tilde{\rho}=\rho_c^*$. Furthermore, the optimal sum rate $\tilde{R_1}+\tilde{R_2}$ is upper-bounded by $\varphi^{-1}(\psi(a))$.
\end{remark}

\begin{remark}\label{remark_linear_simultaneous}
Mostly in literature, the decoding cost is assumed a convex function of the form $\varphi(R)=\beta(2^{2R}-1)$. Moreover, as a worst case of a practical PS-SWIPT systems, $N$ is neglected as $N\ll N_p$ since the splitted signal for the information recovery passes through several standard operation to be converted from the RF band to baseband. During this process, additional noise called signal processing noise is present due to the  phase-offsets,  non-linearities, etc. \cite{clerckx2018fundamentals}.
\end{remark}

As a special case as described in  Remark \ref{remark:psi_linearity} and \ref{remark_linear_simultaneous}, we assume   $\psi(P_{in})=\eta P_{in}$,  $\varphi(R)=\beta(2^{2R}-1)$, and  $N\ll N_p$. Then, 
 the optimal PS factor in $\mathbf{P} \mathbf{1}$ is derived as 
 \begin{equation*}\label{sum-rate-optimal}
   \tilde{\rho} =\frac{\beta (a-N)}{\beta (a-N)+\eta a N_p},
\end{equation*}which is an unit-less scalar ranged between $0$ and $1$.

\subsection{Successive Interference Cancellation Decoding Scheme}
Assume that the destination applies SIC technique to decode the messages received by the users, \cite{el2011network}. To decode the first user's message, the destination treats the second user's signal as a noise. So, the receiver cancels out the recovered signal of the first user and decodes the second user's signal. Accordingly, the achievable rate region is confined by
\begin{align}     
		R_1 &\leq \frac{1}{2}\log\Big(1+\dfrac{(1-\rho)|h_1|^2P_1}{(1-\rho)N_u+N_p}\Big),\label{successr1}\\
		 R_2 &\leq \frac{1}{2}\log\Big(1+\dfrac{(1-\rho)|h_2|^2P_2}{(1-\rho)N+N_p}\Big),	\label{successr2}\\
	\hspace{-1.5cm} \varphi(R_1)+\varphi(R_2) &\leq \psi\big( \rho(|h_1|^2P_1+|h_2|^2P_2+N)\big).\label{successphi}
		\end{align}
		where $N_u=|h_2|^2P_2+N$, and  (\ref{successphi}) indicates the decoding cost restriction resulting from the sequential decoding in SIC.

\begin{theorem}
	The MDRB is calculated as
\begin{equation}
			\begin{aligned}\label{simulregion1}
			\hspace{-0.2cm} \begin{cases}

			 R_1^*=\varphi^{-1}\Big(\psi\big(\rho a\big)-\varphi(R_2^*)\Big),\\ \hspace{5.5cm}\rho_1^*\leq\rho\leq \rho_c^* \\
			 R_2^*=\frac{1}{2}\log\Big(1+\dfrac{(1-\rho)|h_2|^2P_2}{(1-\rho)N+N_p}\Big),
			
						\end{cases}
			\end{aligned}
			\end{equation}
and
\begin{equation}
		\begin{aligned}\label{simulregion3}
		\hspace{-0.2cm}\begin{cases}
			 R_2^*=\varphi^{-1}\Big(\psi\big( \rho a \big)-\varphi(R_1^*)\Big), \\ \hspace{5.5cm} \rho_2^*\leq\rho\leq \rho_c^*\\
			 R_1^*=\frac{1}{2}\log\Big(1+\dfrac{(1-\rho)|h_1|^2P_1}{(1-\rho)N_u+N_p}\Big),
					\end{cases}
			\end{aligned}
\end{equation}
where 
	
	{{\begin{equation}
			\begin{aligned}\nonumber
			&\hspace{-1.5cm}\rho_c^*=\Gamma_c^{-1}( a), \rho_1^*=\Gamma_1^{-1}( a),  \rho_2^*=\Gamma_2^{-1}( a),
			\end{aligned}
			\end{equation} }}
and 
	{{
			\begin{equation}
			\begin{aligned}\nonumber
			\Gamma_c(x)=&\dfrac{1}{x}\psi^{-1}\Bigg(\varphi\left(\frac{1}{2}\log\left(1+\dfrac{(1-x)|h_1|^2P_1}{(1-x)N_u+N_p}\right)\right)\\
			&+\varphi\left(\frac{1}{2}\log\left(1+\dfrac{(1-x)|h_2|^2P_2}{(1-x)N+N_p}\right)\right)\Bigg),\\
			\Gamma_1(x)=&\dfrac{1}{x}\psi^{-1}\Bigg(\varphi\left(\frac{1}{2}\log\left(1+\dfrac{(1-x)|h_2|^2P_2}{(1-x)N+N_p}\right)\right)\Bigg),\\
			\Gamma_2(x)=&\dfrac{1}{x}\psi^{-1}\Bigg(\varphi\left(\frac{1}{2}\log\left(1+\dfrac{(1-x)|h_1|^2P_1}{(1-x)N_u+N_p}\right)\right)\Bigg),\\
			a =&|h_1|^2P_1+|h_2|^2P_2+N.
			\end{aligned}
			\end{equation}}}
	
\end{theorem}
\begin{proof}
	Following the similar steps as in the proof of Theorem 1, the MDRB in SIC is created by the pairs of rates ($R_1^*, R_2^*$) who satisfy (\ref{successphi}) and one of (\ref{successr1}) or (\ref{successr2}) with equality.
\end{proof}
Now, by changing the decoding order, we firstly decode the second user's message. Then, we cancel this data out from the received signal, and finally decode the first user's message. So, by swapping the indices 1 and 2 in the achievable rate region (\ref{successr1})-(\ref{successphi}), we derive the new MDRB $(\hat{R}_1,\hat{R}_2)$. Then, applying the time-sharing technique over the  two MDRBs resulting from each decoding order, the convex hull  of $(R_1^*,R_2^*)\cup (\hat{R}_1,\hat{R}_2)$ forms the maximum departure region.\par
It is worth pointing out that the destination in SIC scheme prefers to primarily decode the data from the user with better channel gain, and thereafter, to decode the other data.  So without loss of generality, we assume that the channel between the first user and the destination is stronger than the second one. Thus, the sum rate optimization problem is written as
\begin{align*}
\hspace{-1.7cm}\mathbf{P} \mathbf{2}:\max _{{\rho,R_1,R_2}}\hspace{0cm} & R_1+R_2 \\
\hspace{-0.5cm}\text { {s.t.} }\hspace{0cm}&(\ref{successr1}),  (\ref{successr2}), (\ref{successphi}), \hspace{0.1cm} \text{and} \hspace{0.1cm}0\leq\rho\leq1, 
\end{align*}
which is a non-convex problem whose solution is not as straightforward as $\mathbf{P} \mathbf{1}$, due to the unfavorable left term in  (\ref{successphi}). However, using the fact that the both functions $\varphi(.)$ and $\psi(.)$ are non-decreasing, we characterize the following lemma.

\begin{lemma}\label{SIC_sumrate_lemma}
Optimally in $\mathbf{P} \mathbf{2}$, the harvested energy at the destination is fully consumed, i.e., (\ref{successphi}) is satisfied with equality. Moreover, at least one of the constraints (\ref{successr1}) or  (\ref{successr2}) are satisfied with equality.
\end{lemma}

\begin{proof}
Assume that (\ref{successphi}) is satisfied with strict inequality in the optimal solution of $\mathbf{P} \mathbf{2}$. The contradiction is proven by decreasing  $\rho$ with small enough $\epsilon>0$ and increasing $R_1$ and $R_2$ with $\epsilon/M$ for a large enough $M>0$. Then, noting that (\ref{successphi}) is satisfied with equality, we assume that both (\ref{successr1}) and (\ref{successr2}) are satisfied with strict inequality in the optimal solution. The contradiction here is shown by feasibly adding $\rho$, $R_1$, and $R_2$ with $\epsilon$ and $\epsilon/M$, respectively. So, Lemma \ref{SIC_sumrate_lemma} follows.
\end{proof}
Lemma \ref{SIC_sumrate_lemma} implies that in the optimal solution of $\mathbf{P} \mathbf{2}$, the both constraints (\ref{successr1}) and (\ref{successr2}) are not necessarily satisfied with equality at the same time. As a result, the optimal solution to $\mathbf{P} \mathbf{2}$ is derived through Theorem \ref{SIC_sumrate_Theorem}.
\begin{theorem}\label{SIC_sumrate_Theorem}
In $\mathbf{P} \mathbf{2}$, the optimal PS factor $\tilde{\rho}$ is calculated as
\begin{equation}\small
\begin{aligned}
\hspace{-0cm} 
\tilde{\rho}=&ARGMAX\big(f_1(\rho_{1}^*),f_2(\rho_{2}^*), f_1(\rho_{11}),..., f_1(\rho_{1n}),\\
&f_2(\rho_{21}),..., f_2(\rho_{2m}),f_1(\rho_{c}^*),f_2(\rho_{c}^*)\big),  
\end{aligned}
\end{equation}
and the optimal sum-rate is derived by
\begin{equation}
\begin{aligned}
\hspace{-0.5cm}  
\tilde{R}_1+\tilde{R}_2=&MAX\big(f_1(\rho_{1}^*),f_2(\rho_{2}^*),f_1(\rho_{c}^*),f_2(\rho_{c}^*),\\
&f_1(\rho_{11}),..., f_1(\rho_{1n}),  f_2(\rho_{21}),..., f_2(\rho_{2m})\big),
\end{aligned}
\end{equation}
where $\rho_{1}^*<\rho_{1i}<\rho_{c}^*$ and $\rho_{2}^*<\rho_{2j}<\rho_{c}^*$, for $1 \leq i\leq n$ and $1 \leq j\leq m$, denote the $i^{\text{th}}$ and $j^{\text{th}}$  answers to the differential equations $f'_1(x)=0$ and $f'_2(x)=0$, respectively, wherein
\begin{align*}
\hspace{-1.7cm} & f_1(x)=\varphi^{-1}\Big(\psi\big(x a\big)-\varphi\big(g_1(1-x)\big)\Big)+g_1(1-x),\\
\hspace{-0.5cm} & f_2(x)=\varphi^{-1}\Big(\psi\big(x a\big)-\varphi\big(g_1(1-x)\big)\Big)+g_2(1-x)\\
&g_1(x)=\frac{1}{2}\log\big(1+\dfrac{x|h_2|^2P_2}{xN+N_p}\big),\\
&g_2(x)=\frac{1}{2}\log\big(1+\dfrac{x|h_1|^2P_1}{xN_u+N_p}\big).
\end{align*}
\end{theorem}
\begin{proof}
See Appendix A.
\end{proof}
\begin{lemma}\label{SIC_linear_optimal}
As a special case as described in  Remark \ref{remark:psi_linearity} and \ref{remark_linear_simultaneous}, we assume   $\psi(P_{in})=\eta P_{in}$,  $\varphi(R)=\beta(2^{2R}-1)$, and  $N\ll N_p$. Then, the optimal PS factor in $\mathbf{P} \mathbf{2}$ is calculated as
\begin{equation}
\tilde{\rho}=0.5\left(1+ \dfrac{\beta B^2+\beta C N_p+\eta a N_p^2-\Delta}{\beta B^2+\eta a N_p B}\right),
\end{equation}
and the optimal sum rate is derived as
\begin{equation}
\begin{aligned}
\hspace{-8cm}\tilde{R}_1+\tilde{R}_2=&\dfrac{1}{2}\log\Big(1+\dfrac{(1-\tilde{\rho})A}{(1-\tilde{\rho})N_u+N_p}\Big)\\ &+\frac{1}{2}\log\Big(1+\dfrac{(1-\tilde{\rho})B}{N_p}\Big)
\end{aligned}
\end{equation}
where
\begin{equation}
   \hspace{-0cm} \begin{aligned} \label{tarif}
    A=&|h_1|^2P_1,~~~ B=|h_2|^2P_2,~~~ C=A+B,\\
    \Delta=&N_p\big[(\eta a B-\beta C)^2 + \eta a (\eta a N_p^2+2\eta a N_p B
    \\&+2\beta C N_p+4\beta B^2)\big]^{0.5}  .
    \end{aligned}
\end{equation}

\end{lemma}
\begin{proof}
From Lemma \ref{SIC_sumrate_lemma} and Appendix B.
\end{proof}
\section{PS-SWIPT MAC with user cooperation}
\subsection{System Model}
Here, we consider a two-user PS-SWIPT MAC with user cooperation wherein the destination acquires energy for decoding data by applying PS scheme over the received signals transmitted from both users. As depicted in Fig. 2, users potentially enhance the departure region by  creating a common information (see \cite{arafa2016energy}, \cite{su2015cooperative}) 
subject to the decoding cost at all receiving nodes, i.e., both users and the destination. Note that in this work, the terms user and encoder are used interchangeably. Also, perfect CSI is considered.\par
By utilizing the well-known Block Markov Encoding (BME)  for MAC with cooperative encoders at each transmission time block, both encoders transmit their codewords not only based on their own message, but also based on the decoded message of the other encoder which is decoded in the previous block, see \cite{el2011network}, \cite{kaya2007power} and \cite{willems1985discrete}. So, at each transmission block, both encoders transmit common and fresh messages, while the destination constructively receives the common message transmitted from both encoders. Then, since the common message is known at the encoders, each encoder is able to decode the fresh part of the other encoders’ message. This procedure is possible thanks to BME and
backward decoding at the destination which is developed for the classical relay channel, and for the MAC with noisy feedback. So, the transmitted codeword by each user is given by
\begin{align*}
\hspace{-0.3cm}X_i=\sqrt{P_{ij}}X_{ij} + \sqrt{P_{ui}}X_u, \hspace{0.5cm} i=1, 2, \hspace{0.2cm} j=1, 2, \hspace{0.2cm} i\neq j,
\end{align*}
\begin{figure}[t]
\centering
	\includegraphics[scale=0.65]{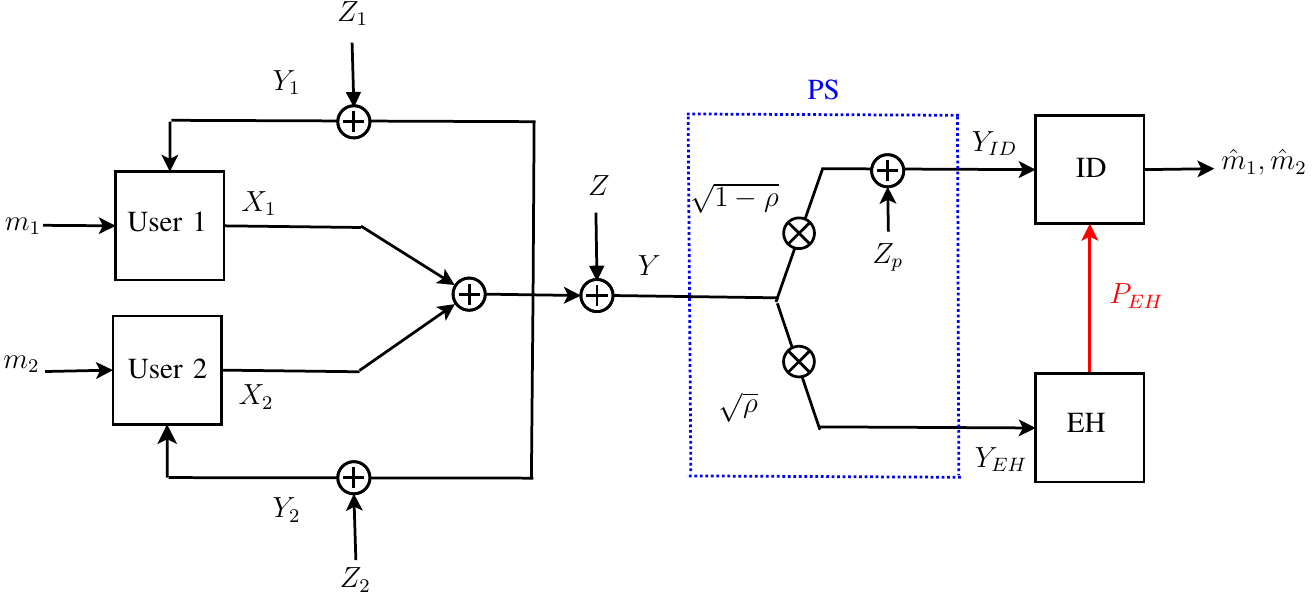}
\caption{PS-SWIPT MAC with user cooperation.}\label{a}
\end{figure}
where $X_{ij}$ with  $\mathbb{E}\lbrace{\vert{X_{ij}}\vert}^{2}\rbrace=P_{ij}$  denotes the transmitted symbol of the fresh information at user $i$ to be decoded at user $j$ and to create the symbol $X_{u}$, the common information with power $P_{ui}$. Therefore, $P_i=P_{ij}+P_{ui}$ refers to the total transmission power at user $i$. Now, the received signals at users and the destination are given by 
\begin{align*}
&\hspace{0cm}Y_j=h_{ij}X_{i} + Z_j,
\hspace{1cm} i=1, 2, \hspace{0.2cm} j=1, 2, \hspace{0.2cm} i\neq j, \\
&\hspace{0cm}Y=h_1X_1 + h_2X_2 + Z,
\end{align*}
where $h_j$ indicates the channel gain between user $j$ and the  destination and $h_{ij}$ shows the channel gain from user $i$ to $j$.
Besides, $Z_j \sim \mathcal{CN}(0,N_j)$ and $Z \sim \mathcal{CN}(0,N)$ represent AWGN terms at user $j$ and the destination, respectively.
However, it is assumed that the user-user links are stronger than user-destination one, i.e., $N\gg N_j$. Similarly, by applying the PS method, the received signal at the destination is splitted into two different signals through the PS block. Thus, the signals $Y_{EH}$ for EH and $Y_{ID}$ for ID are formulated as 
\begin{align}
&Y_{EH}=\sqrt{\rho}Y,\\
&Y_{ID}=\sqrt{1-\rho}Y+Z_{P},\\
&P_{EH}=\psi\Big(\rho\big(S+N\big)\Big),
\end{align}
where $S=|h_1|^2P_1+|h_2|^2P_2+2h_1h_2\sqrt{P_{u1}P_{u2}}$ refers to the received RF signal power,
$P_{EH}$ is the total harvested power,
 and $Z_p \sim \mathcal{CN}(0,N_p)$ denotes the PS signal processing noise.

\subsection{Maximum Departure Region and Optimal Power Allocation}
To obtain the maximum departure region with the set of powers $\boldsymbol{P}=\{P_{12},P_{21},P_{u1},P_{u2}\}$, a pair of achievable rates $(R_1,R_2)$ satisfies
\begin{align}
 R_1 &\leq \frac{1}{2}\log\left(1+\dfrac{|h_{12}|^2P_{12}}{N_2}\right),\label{eqc1}\\
 R_2 &\leq  \frac{1}{2}\log\left(1+\dfrac{|h_{21}|^2P_{21}}{N_1}\right),\label{eqc2}\\
 R_1+R_2 &\leq \frac{1}{2}\log\left(1+\dfrac{(1-\rho)S}{(1-\rho)N+N_p}\right),\label{eqc3}\\
   R_1+R_2 &\leq \varphi^{-1}\left(P_{EH}\right),\label{eqc4}\\
    \varphi_1(R_2) + P_1 &\leq P_{U1}, \label{eqc5}\\ 
 \varphi_2(R_1)+P_2&\leq P_{U2}. \label{eqc6}
\end{align}
The inequalities (\ref{eqc1})-(\ref{eqc3}) denote the achievable rate constraints at PS-SWIPT MAC with user cooperation and (\ref{eqc4})-(\ref{eqc6}) represent the decoding cost constraints at the receiver nodes. The decoding cost function $\varphi_1(.)$ and $\varphi_2(.)$ at users are defined similar to $\varphi(.)$ at the destination and the pair of $P_{U1}$ and $P_{U2}$ refers to the amount of power available at user 1 and user 2, respectively.\par
 In this section, we firstly characterize the MDRB by forming the weighted sum rate as 
$$R_\mu=\mu_1R_1 + \mu_2R_2 ,$$
 subject to all the constraints (\ref{eqc1})-(\ref{eqc6}), \cite{su2015cooperative}. As proven in\cite[Lemma~1]{arafa2016energy}, there exists an optimal policy to achieve the MDRB, where (\ref{eqc1}) and (\ref{eqc2}) hold with equality. Therefore, for any $\mu_1,\mu_2\geq0$, the optimization problem is formulated as 
\begin{align}
\hspace{-0.2cm}\mathbf{P}\mathbf{3}: \max _{\rho, \boldsymbol{P}} \hspace{0.2cm} & \mu_1g\left(bP_{12}\right)+\mu_2g\left(cP_{21}\right)  \label{os_general}\\
\text { {s.t.} }\hspace{0.2cm} &\varphi_1\big(g( cP_{21})\big)+P_{12}+P_{u1}\leq P_{U1} , \label{os1_general} \\
& \varphi_2\big(g(bP_{12})\big)+P_{21}+P_{u2}\leq P_{U2} , \label{os2_general}\\
&g\left(bP_{12}\right)+g(cP_{21}) \leq \varphi^{-1}\Big(\psi \big(\rho(S+N) \big)\Big),\label{os3_general}\\
& g\left(bP_{12}\right) + g\left(cP_{21}\right)\leq\  g\Big(\dfrac{(1-\rho)S}{(1-\rho)N+N_p}\Big),\label{os4_general}\\
& 0\leq\rho\leq1 ,\\
& P_{12},P_{21},P_{u1},P_{u2} \geq 0.
\end{align}
where $g(x)\triangleq \frac{1}{2} \log(1+x)$, $b=\frac{|h_{12}|^2}{N_2}$, and $c=\frac{|h_{21}|^2}{N_1}$.
\begin{lemma}\label{t1-general}
In the optimal solution of $\mathbf{P} \mathbf{3}$, all the constraints (\ref{os1_general})-(\ref{os4_general}) hold with equality. 
\end{lemma}
\begin{proof}
On the user side, suppose some energy is left at the first user in the optimal solution and thus, (\ref{os1_general}) holds with strict inequality. So, we can increase $P_{u1}$ with a small enough value $\epsilon>0$ such that (\ref{os1_general}) is still  satisfied with strict inequality. Noting that the functions $\varphi(.)$, $\varphi_1(.)$, $\varphi_2(.)$, $\psi(.)$, and $g(.)$ are all increasing and $S$ includes an addition term as a function of the multiplication of $P_{u1}$ and $P_{u2}$, (\ref{os3_general}) and (\ref{os4_general}) are also satisfied with strict inequality. 
Then, there exist a large enough value $M>0$ to add and subtract $P_{21}$ and $P_{u2}$ with $\epsilon/M$, respectively,  such that $S$ increases and the left term of (\ref{os2_general}) is kept unchanged, while (\ref{os1_general}), (\ref{os3_general}), and (\ref{os4_general}) are not violated. This act increases $P_{21}$, and thus, the objective function which contradicts the optimality. Therefore, (\ref{os1_general}) is satisfied with equality. Similarly, we can prove that (\ref{os2_general}) is also satisfied with equality. On the destination side, we assume that (\ref{os3_general}) is satisfied with strict inequality. Then, for a small enough $\epsilon$ and large enough $M$, we can decrease $\rho$, $P_{u1}$, and $P_{u2}$ with  $\epsilon$, $\epsilon/M$, and $\epsilon/M$, respectively, such that the right term of (\ref{os4_general}) increases and (\ref{os3_general}) is still satisfied with strict inequality. Now, all the constraints (\ref{os1_general})-(\ref{os4_general}) are satisfied with strict inequality. Then, there is large enough $N>0$ such that adding the both $P_{21}$ and $P_{12}$ with $\epsilon/MN$ does not violate non of the constraints (\ref{os1_general})-(\ref{os4_general}), showing the contradiction. Therefore, (\ref{os3_general}) is also satisfied with equality in the optimal solution. Finally, assuming that (\ref{os4_general}) is satisfied with strict inequality, we increase $\rho$ and decrease $P_{u1}$ and $P_{u2}$ with $\epsilon$,  $\epsilon/M$, and $\epsilon/M$, respectively to show the contradiction. As a result, (\ref{os4_general}) is also satisfied with equality in the optimal solution of  $\mathbf{P} \mathbf{3}$.
\end{proof}
It is worth noting that decreasing $P_{u1}$ and $P_{u2}$ in the proof of Lemma \ref{t1-general} is feasible as long as 
$P_{u1},P_{u2}> 0$. Otherwise, $P_{u1}^*=P_{u2}^*= 0$ in the optimal solution of $\mathbf{P} \mathbf{3}$, or in other words, it is not optimal to employ cooperation anymore. As a result of  Lemma \ref{t1-general}, we solve
$\mathbf{P} \mathbf{3}$ using the non-linear system of four equations in (\ref{os1_general})-(\ref{os4_general}) and four unknowns in $\boldsymbol{P}$ as a function of $\rho$. Then, we equate the derivative of the objective function in terms of $\rho$ with zero, yielding the the optimal PS factor $\rho^*$. Accordingly, the optimal set of powers $P_{12}^*$, $P_{21}^*$,$P_{u1}^*$, and $P_{u2}^*$ are derived. However, to find a closed form solution, we reduce $\mathbf{P} \mathbf{3}$ into
\begin{align}
\hspace{-0.2cm}\mathbf{P}\mathbf{4}: \max _{\rho, \boldsymbol{P}} \hspace{0.2cm} & \mu_1g\left(bP_{12}\right)+\mu_2g\left(cP_{21}\right)  \label{os}\\
\text { {s.t.} }\hspace{0.2cm} &\beta cP_{21}+P_{12}+P_{u1}\leq P_{U1} , \label{os1} \\
& \beta bP_{12}+P_{21}+P_{u2}\leq P_{U2} , \label{os2}\\
&\beta\Big(bP_{12}+cP_{21}+bcP_{12}P_{21}\Big) \leq \eta\rho(S+N) ,\label{os3}\\
& g\left(bP_{12}\right) + g\left(cP_{21}\right)\leq\  g\Big(\dfrac{(1-\rho)S}{(1-\rho)N+N_p}\Big),\label{os4}\\
& 0\leq\rho\leq1 ,\\
& P_{12},P_{21},P_{u1},P_{u2} \geq 0.
\end{align}wherein the exponential functions $\varphi(.)=\varphi_1(.)=\varphi_2(.)=\beta(2^{2R}-1)$ and a linear EH model $\psi(P_{in})=\eta P_{in}$ as in described in Remark \ref{remark:psi_linearity} and \ref{remark_linear_simultaneous} are assumed.
Using Lemma (\ref{t1-general}), we solve $\mathbf{P}\mathbf{4}$ by maximizing the objective function subject to (\ref{os1})-(\ref{os4}) satisfied with equality. So,
\begin{align}
&P_{12}= P_{U1} -P_{u1}-\beta c P_{21},\label{p1}\\
&P_{21}= P_{U2}-P_{u2}-\beta bP_{12}.\label{p2}
\end{align} 
By substituting $P_{21}$ and $P_{12}$ in the right terms of (\ref{p1}) and (\ref{p2}), we have
\begin{align}
&P_{12}= \frac{P_{U1}-\beta cP_{U2}-P_{u1}+\beta c P_{u2}}{1-\beta^2bc},\label{p3}\\
&P_{21}= \frac{P_{U2}-\beta bP_{U1}-P_{u2}+\beta b P_{u1}}{1-\beta^2bc}.\label{p4}
\end{align} 
Then, the objective function of $\mathbf{P} \mathbf{3}$ can be rewritten as  
 \begin{equation}\small
\mu_1g\left(b\frac{A-P_{u1}+\beta cP_{u2}}{1-\beta^2bc}\right)+  \mu_2g\left(c\frac{B-P_{u2}+\beta b P_{u1}}{1-\beta^2bc}\right),\label{obj}
\end{equation}
where $A=P_{U1}-\beta cP_{U2}$  and $B=P_{U2}-\beta b P_{U1}$. Now, by taking derivative of (\ref{obj}) with respect to $P_{u1}$ and $P_{u2}$, we have
\begin{align}
bc\Big(\mu_1+\mu_2\beta^2bc\Big)\tilde{P}_{u2} - \beta b^2c\Big(\mu_1+\mu_2\Big)\tilde{P}_{u1}=C_1,\label{pu3}\\
bc\Big(\mu_2+\mu_1\beta^2bc\Big)\tilde{P}_{u1}-\beta bc^2\Big(\mu_1+\mu_2\Big)\tilde{P}_{u2}=C_2, \label{pu4}
\end{align}    
where
\begin{align*}
C_1=\mu_1b\Big(1-\beta^2bc+Bc\Big) - \mu_2  \beta bc\Big(1-\beta^2bc+Ab\Big), \\
 C_2=\mu_2c\Big(1-\beta^2bc+Ab\Big) - \mu_1 \beta bc\Big(1-\beta^2bc+Bc\Big).
\end{align*}

\begin{theorem}\label{cooperative_linear_solution}
In $\mathbf{P}\mathbf{4}$, the optimal points ($\tilde{P}_{u1}$,$\tilde{P}_{u2}$)
and the optimal PS factor $\tilde{\rho}$ can be written as 
\begin{align}\label{P_u1&P_u2}
    \hspace{-2cm}\begin{pmatrix} 
    \tilde{P}_{u1}\\
    \tilde{P}_{u2}
    \end{pmatrix}
=\frac{1}{ED-FC}
\begin{pmatrix}
F&D\\
E&C
\end{pmatrix}
\begin{pmatrix}
C_1\\
C_2
\end{pmatrix},
\end{align}
\begin{align}
  \hspace{-2.5cm}\tilde{\rho}=\frac{\beta\left(b\tilde{P}_{12}+c\tilde{P}_{21}+bc\tilde{P}_{12}\tilde{P}_{21}\right)}{\eta(\tilde{S}+N)},\label{rho1}
\end{align}
where
\begin{align*}
&C=\beta b^2c\Big(\mu_1+\mu_2\Big),\hspace{0.1cm} &&D=bc\Big(\mu_1+\mu_2\beta^2bc\Big),\\
&E=bc\Big(\mu_2+\mu_1\beta^2bc\Big),\hspace{0.1cm} &&F=\beta bc^2\Big(\mu_1+\mu_2\Big),
\end{align*}
\begin{align*}
    \hat{S}=|h_1|^2(\tilde{P}_{12}+\tilde{P}_{u1})+|h_2|^2(\tilde{P}_{21}+\tilde{P}_{u2})+2h_1h_2\sqrt{\tilde{P}_{u1}\tilde{P}_{u2}}.
\end{align*}
Then, the optimal values ({$\tilde{P}_{12}$,$\tilde{P}_{21}$})  are derived by substituting $\tilde{P}_{u1}$ and $\tilde{P}_{u2}$ into  (\ref{p3}) and (\ref{p4}).
\end{theorem}
\begin{proof}
The optimal powers ($\tilde{P}_{u1}$,$\tilde{P}_{u2}$) are derived by solving the system of linear equations (\ref{pu3}) and (\ref{pu4}). Then, the optimal PS factor $\tilde{\rho}$ is obtained  by setting  (\ref{os3}) with equality.
\end{proof}
\begin{remark}
The optimal solution in  Theorem \ref{cooperative_linear_solution} is valid only for   $\tilde{P}_{12},\tilde{P}_{21},\tilde{P}_{u1},\tilde{P}_{u2}\geq 0$. Otherwise, 
it is not optimal to employ cooperation anymore. In that case, the classical PS-SWIPT is the recommended framework to enhance the sum rate. Furthermore, if $\beta^2 bc=1$, then $ED=FC$ and the optimal  powers  ($\tilde{P}_{u1}$, $\tilde{P}_{u2}$, $\tilde{P}_{12}$, $\tilde{P}_{21}$) are not unique.
\end{remark}

\begin{figure}[t]
\centering
\subfloat[$\varphi(R)=\beta(2^{2R}-1)$]{\includegraphics[scale=0.335]{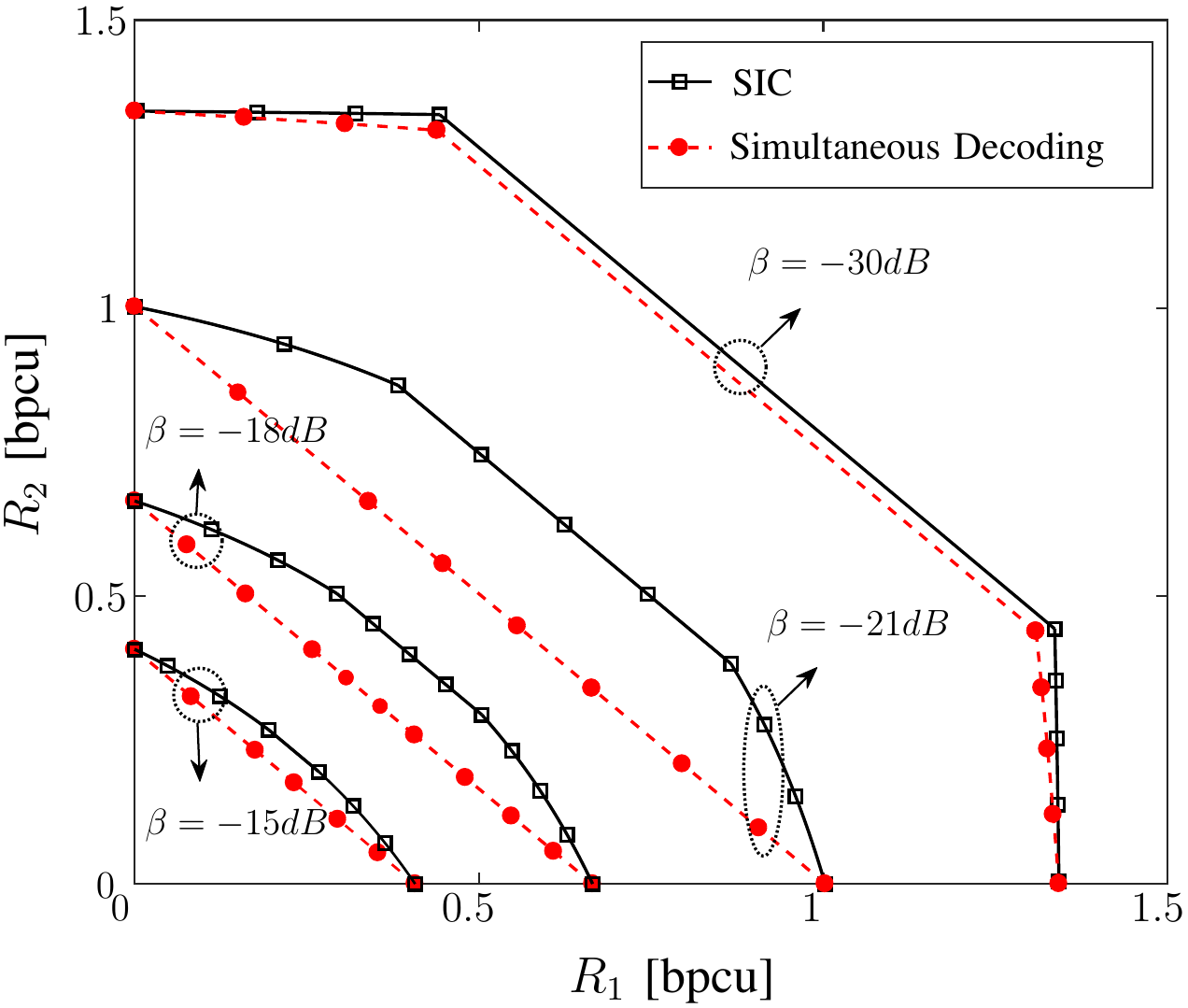}\label{mdrb_convex}}
	\hfil \hspace{0cm}
\subfloat[ $\varphi(R)=\beta\log(2R+1)$]{\includegraphics[scale=0.335]{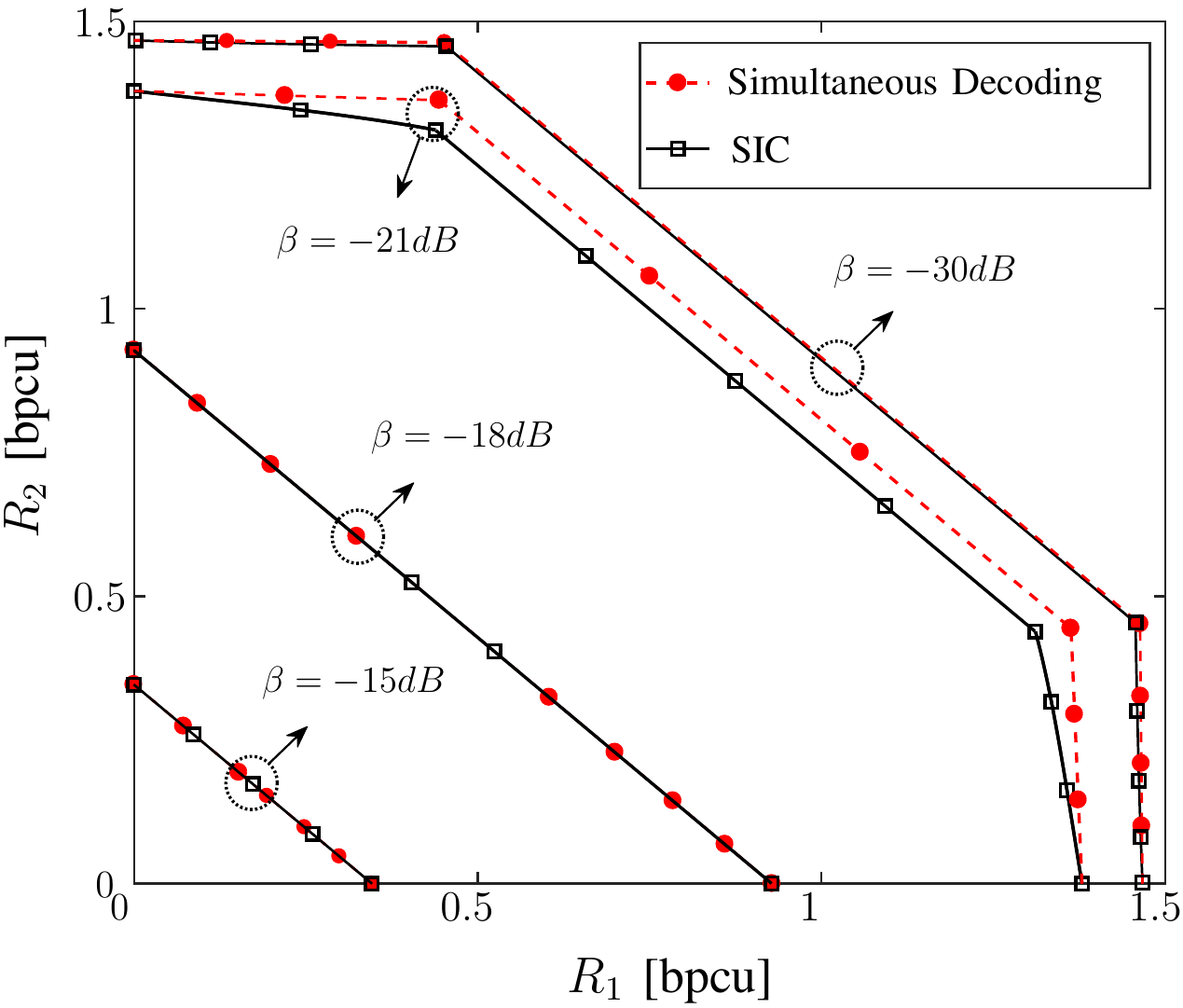}\label{mdrb_concave}}\\
\subfloat[$\varphi(R)=\beta(2R)$]{\includegraphics[scale=0.335]{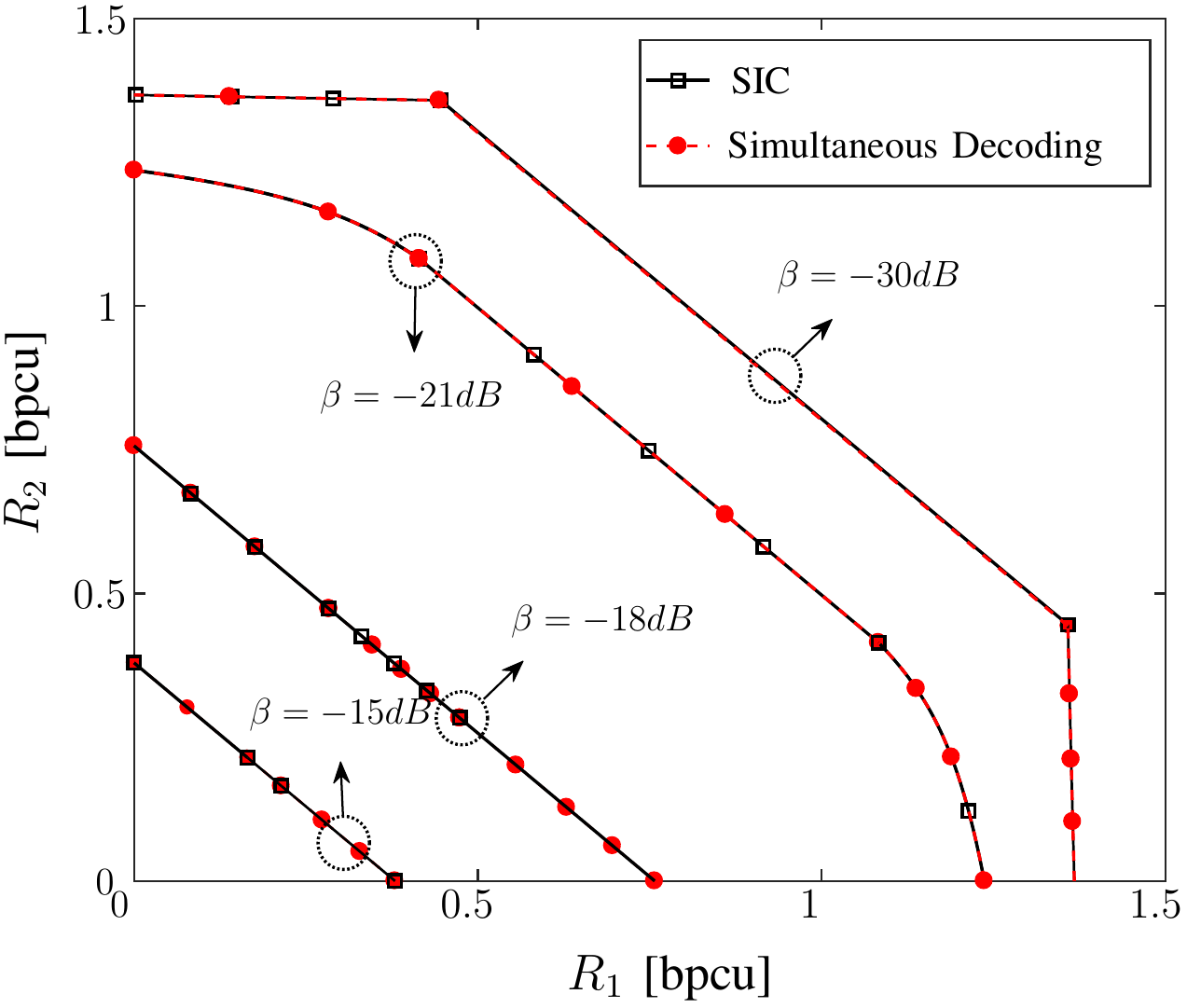}\label{mdrb_linear}}
	\hfil \hspace{0cm}
\subfloat[$\varphi(R)=\varphi_0$]{\includegraphics[scale=0.335]{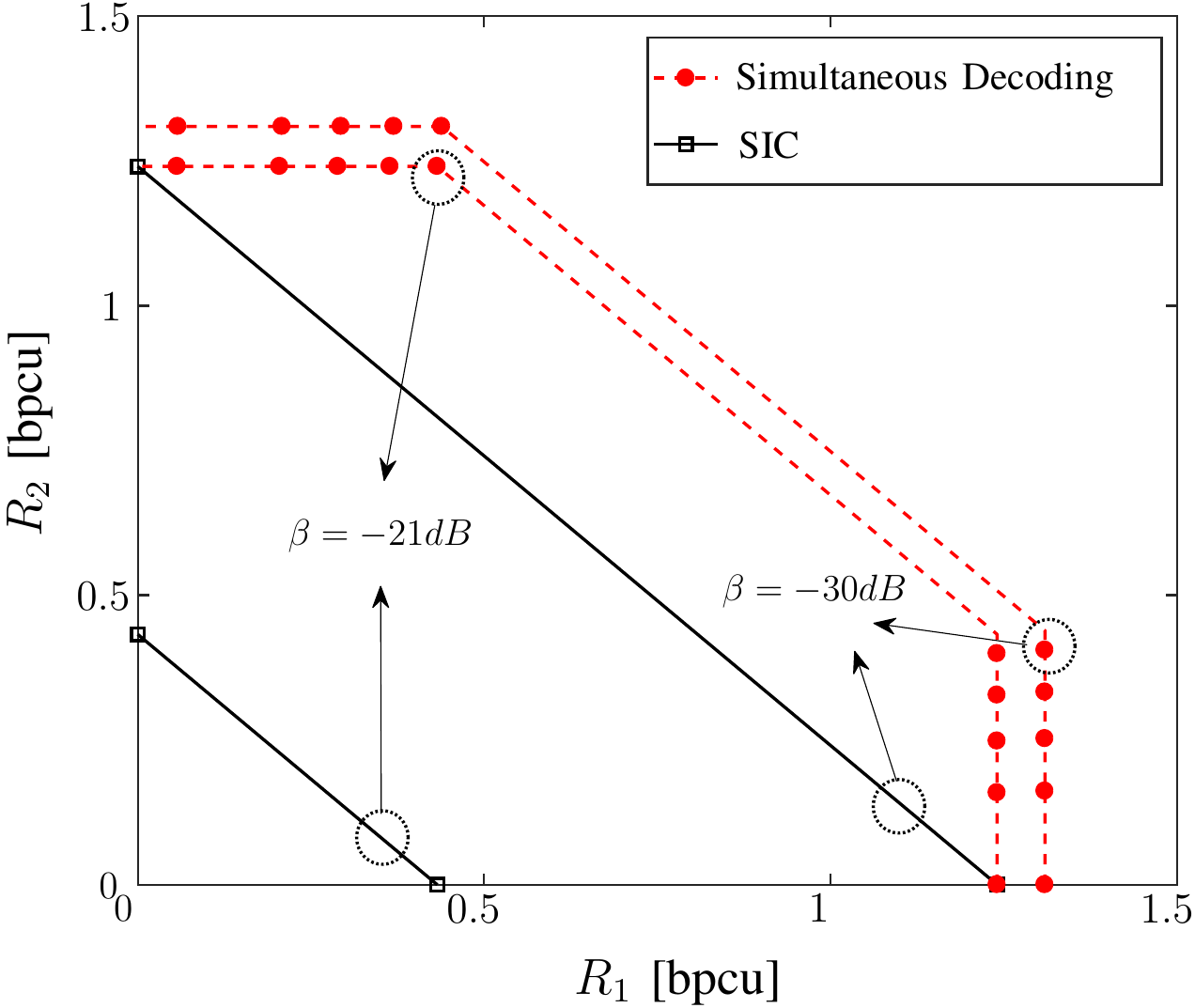}\label{mdrb_constant}}
   \caption{Maximum departure region of the classical PS-SWIPT MAC for simultaneous and SIC decoding schemes and four decoding cost forms.}
  \label{mdrbs} 
\end{figure}

\section{Numerical Results and Discussions}
In this section, we evaluate the performance of PS-SWIPT MAC systems under a non-linear EH model and decoding cost constraints. Specifically, to gather a practical overview of the systems, we consider the non-linear function of the EH rectifier.
To this end, we set the circuity parameters of the EH rectifier in (\ref{rectifier:q1q2}) as $q_1=1500$ and $q_2=0.0022$ with $P_{max}^{DC}=24$mW (see \cite{boshkovska2017max,ma2019generic,boshkovska2016power,xiong2017rate,boshkovska2015practical}).
Furthermore, we assume that the users, with power sources $P_1=P_2=0.5$W, are located at $d_1=d_2=3$m from the destination and the channel gains are defined by $h_1=d_1^{-\alpha}$ and $h_2=d_2^{-\alpha}$, with the propagation exponent $\alpha=2$. Finally, as described in Remark \ref{remark_linear_simultaneous}, we apply the worst case in practical PS-SWIPT by setting $N_p=-30$dB and $N=-60$dB to satisfy $N_p\gg N$. 
\par

Fig. \ref{mdrbs} compares the maximum departure region of the classical PS-SWIPT MAC for simultaneous decoding and SIC schemes with representative decoding cost functions: convex, concave, linear, and constant ones.  Note that the constant decoding cost function might be interpreted as the conventional \emph{signal processing cost} in the literature irrespective of the incoming rate \cite{arafa2017energy} and\cite{orhan2014energy}. As observed in Fig. \ref{mdrb_convex}, the SIC scheme outperforms the simultaneous one with the convex decoding cost, as a smaller amount of power is consumed for decoding the same pair of rates. 
While setting a concave decoding cost function, Fig. \ref{mdrb_concave} indicates that the simultaneous decoding outperforms  the SIC scheme.
However, Fig. \ref{mdrb_linear} shows that both the decoding schemes display the same performance with linear decoding cost function. Finally, setting constant values for the decoding cost as shown in Fig. \ref{mdrb_constant}, the cost for the SIC scheme is doubled compared to the simultaneous one, due to the sequential decoding manner in SIC. 
Here, if $\varphi_0>\psi(a)/2\simeq P_{max}^{DC}/2= 12$mW, the rate region significantly shrinks in SIC as the destination carries enough power to decode only one of the user's messages. Moreover, setting $\varphi_0>\psi(a)\simeq P_{max}^{DC}=24$mW results in an empty rate region for both the decoding schemes, as the EH block lacks enough power for the decoding process.

\begin{figure}[t!]
	\centering
\subfloat[Simultaneous Decoding]{\includegraphics[scale=0.335]{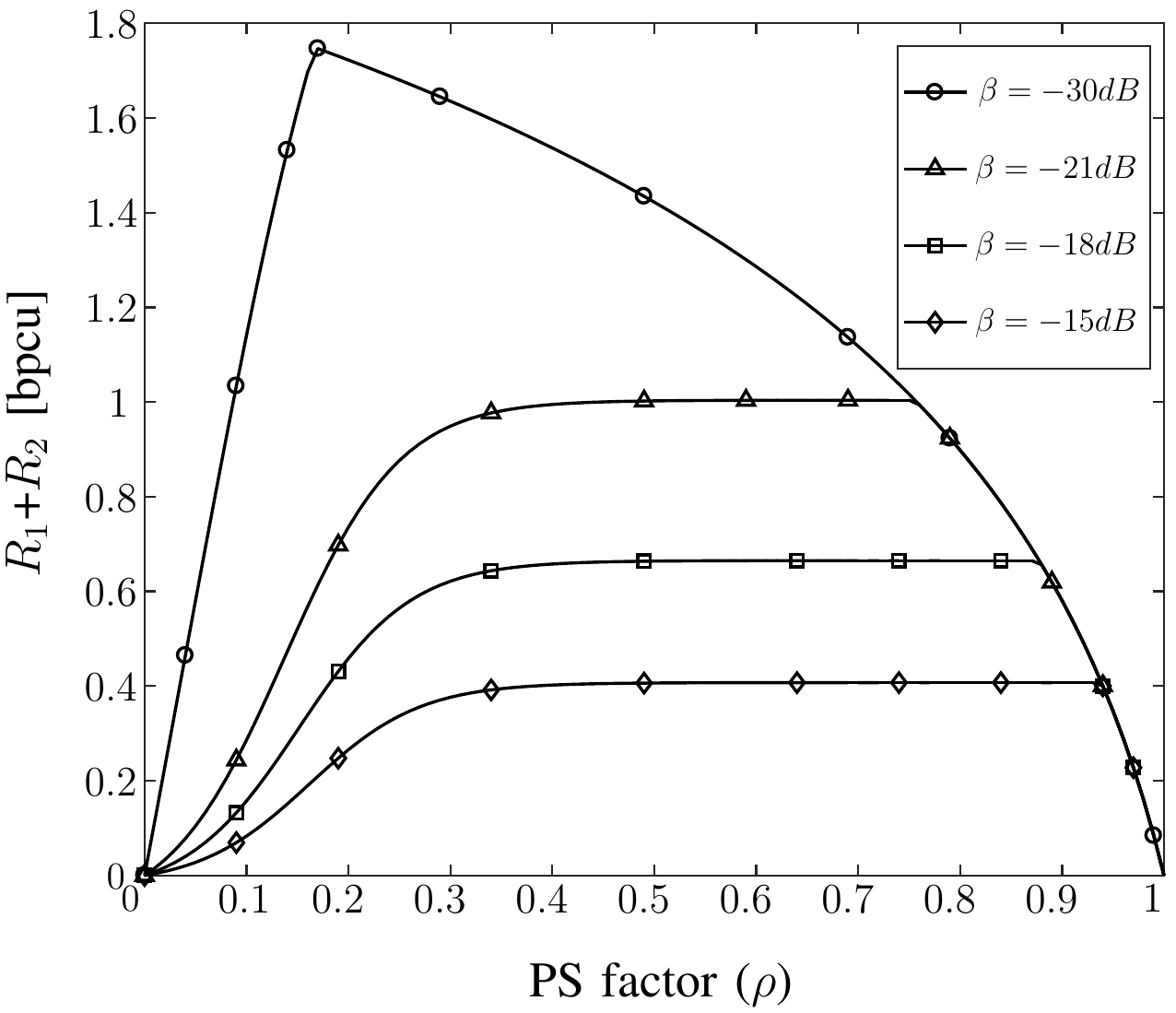}\label{rho_sd_convex}}
	\hfil \hspace{0cm}
\subfloat[SIC]{\includegraphics[scale=0.335]{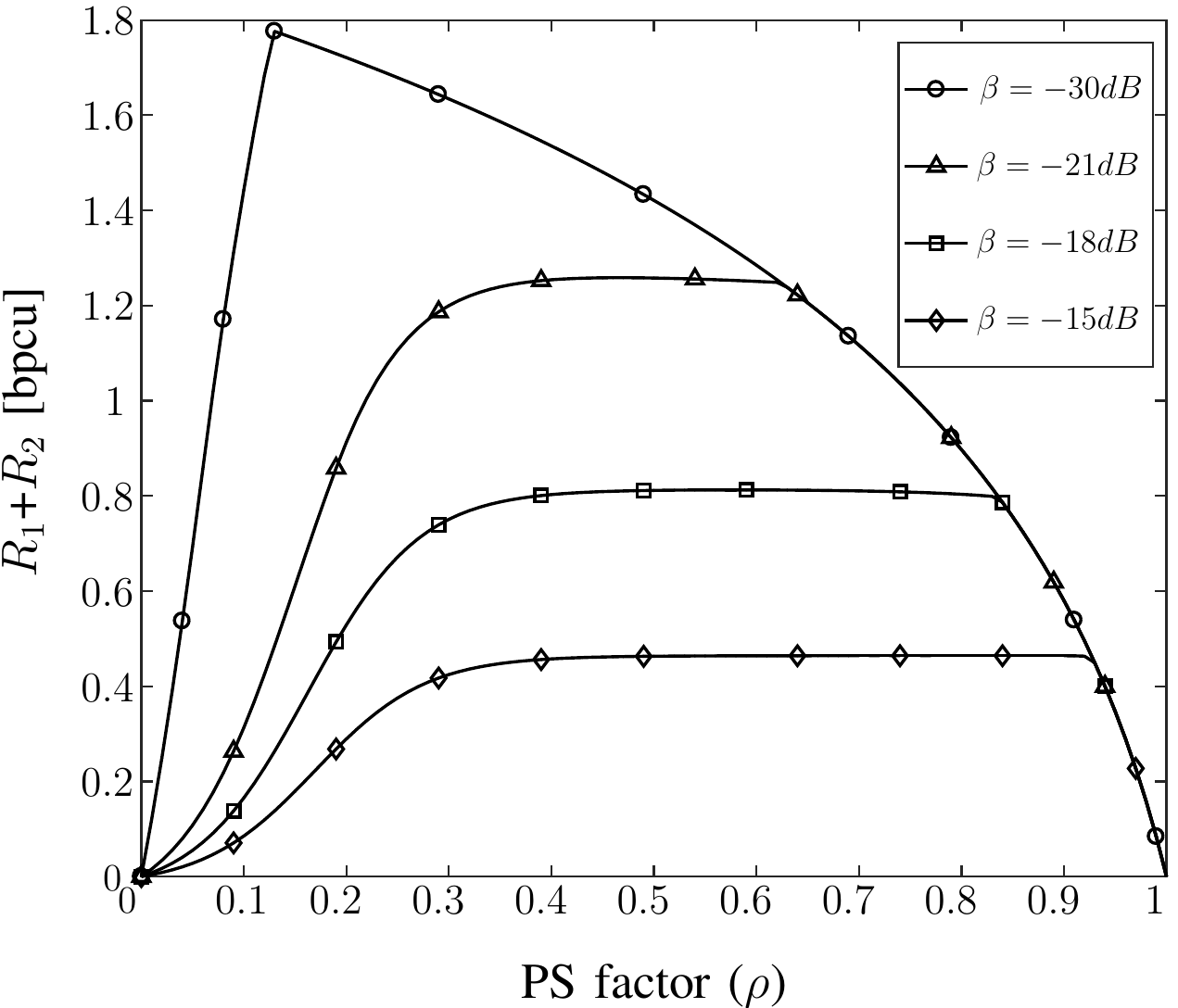}\label{rho_sic_convex}}
	\caption{ Maximum achievable sum rate in terms of the PS factor in Classical PS-SWIPT MAC with the decoding cost function $\varphi(R)=\beta(2^{2R}-1)$.}
	\label{rho_sumrate}
	
\end{figure}

 Fig. \ref{rho_sumrate} illustrates the maximum achievable sum rate in terms of the PS factor in the Classical PS-SWIPT MAC system with the convex decoding cost function for both the decoding schemes. It is observed that increasing $\rho$ saturates 
 the EH rectifier for a large $\beta$, which results in an upper-bound maximum achievable sum rate at which the bound is decreasing by increasing $\beta$ in both the decoding schemes. Moreover, in case the decoding factor is as small as $\beta=-30dB$, the necessary power for decoding is much less than $\psi(a)\simeq P_{max}^{DC}$. As a result, the EH rectifier approximately performs within the linear functionality.
 \par

Finally, Fig. \ref{mdrb_sdcoop} investigates the effect of the  user-user channel gain values and $\beta$  in the performance of the classical PS-SWIPT MAC and PS-SWIPT MAC with cooperation. In the cooperation case, we set the user-user channel parameters as $N_1=N_2=N$ and $h_{12}=h_{21}=h_u$, and the channel between the users and destination similar to the classical case. Further, we assume $\varphi_1(R)=\varphi_2(R)=\varphi(R)$ as the convex decoding cost function. It is observed that a large $\beta$ suppresses the maximum departure region in both the classical and cooperation systems irrespective of the channel gain. Therefore, although the cooperation outperforms the classical system for $h_u\geq 0.004$, this superiority fades as $\beta$ increases such that by setting $\beta=-21$dB, the system does not benefit from the cooperation among the users anymore. On the other hand, if the user-user channel gain drops to a low value as $h_u< 0.004$,  the classical PS-SWIPT MAC outperforms the cooperation system, regardless of $\beta$. 

\begin{figure}[t]
\centering
\subfloat[$\beta=-30 dB$]{\includegraphics[scale=0.335]{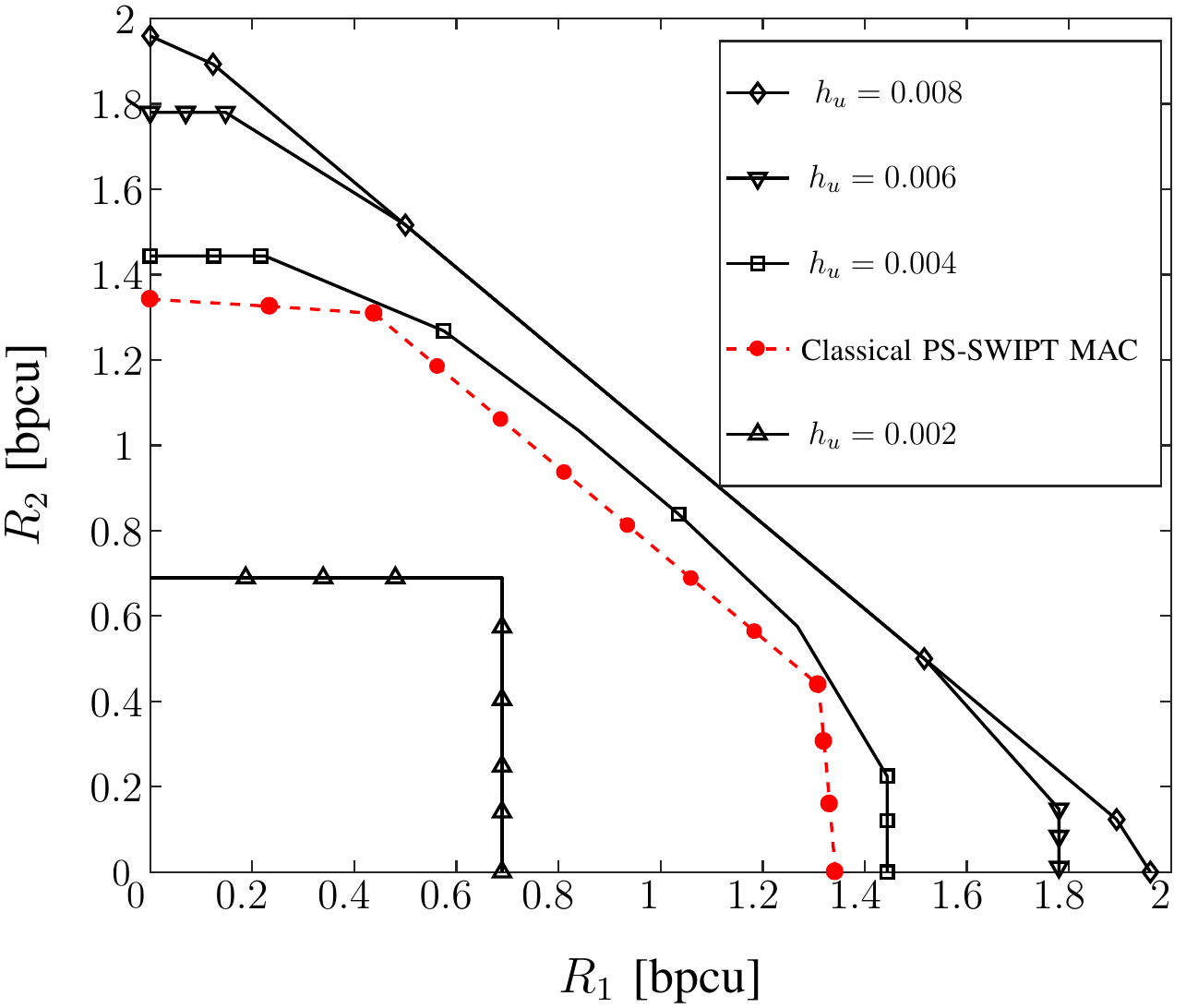}\label{mdrb_sdcoop_h30}}
	\hfil \hspace{0cm}
\subfloat[$\beta=-27 dB$]{\includegraphics[scale=0.335]{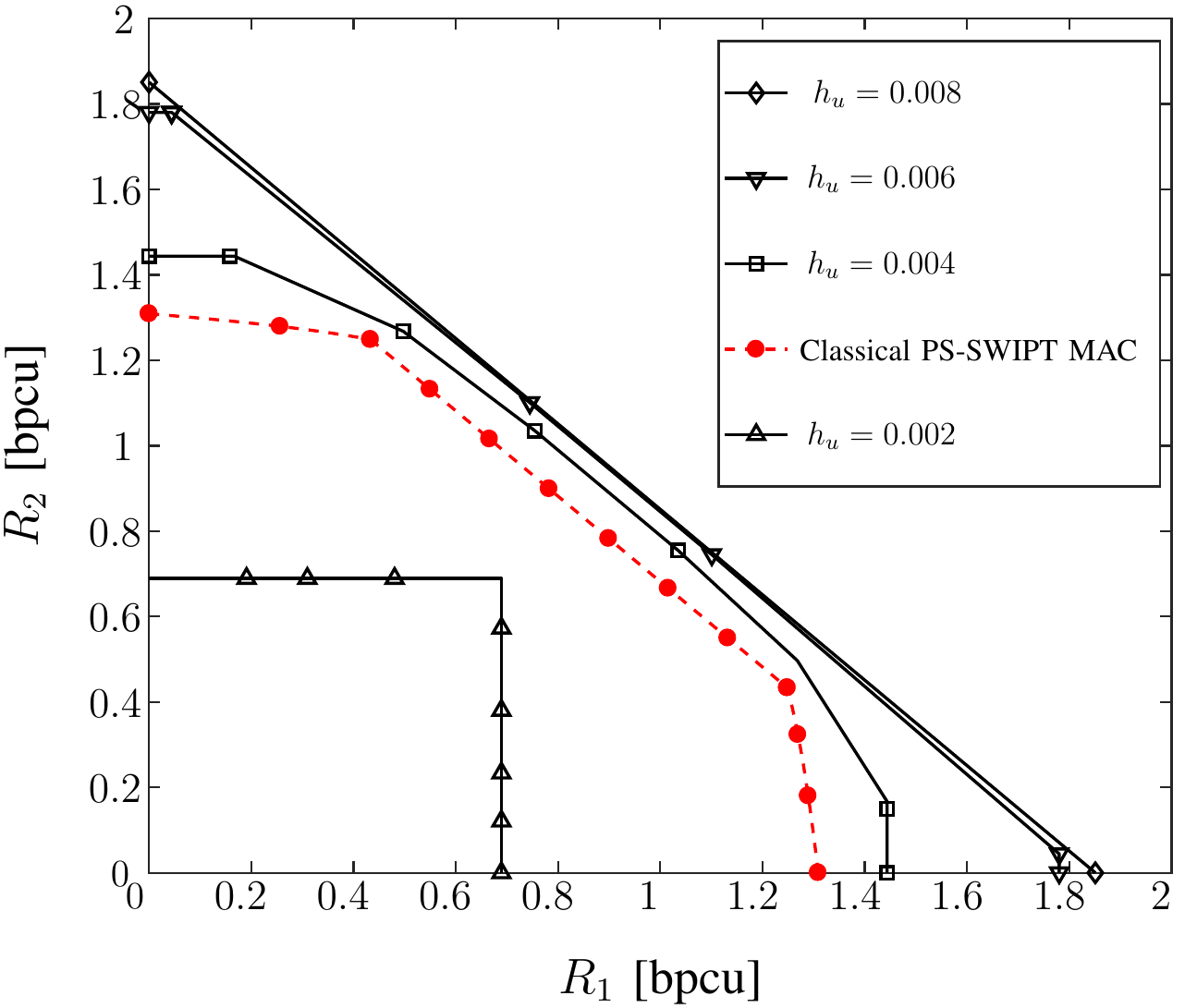}\label{mdrb_sdcoop_h27}}\\
\subfloat[$\beta=-24 dB$]{\includegraphics[scale=0.335]{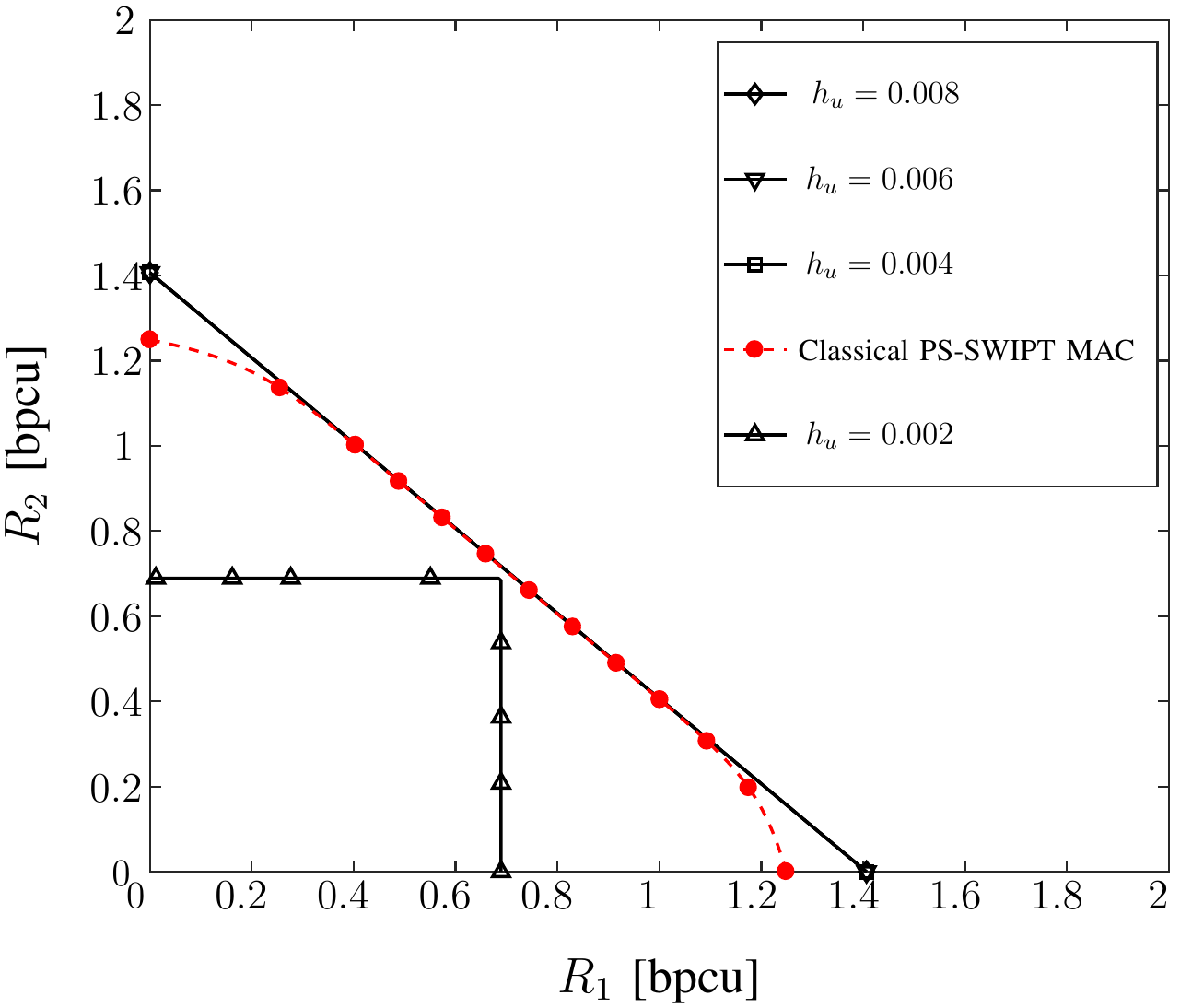}\label{mdrb_sdcoop_h24}}
	\hfil \hspace{0cm}
\subfloat[$\beta=-21 dB$]{\includegraphics[scale=0.335]{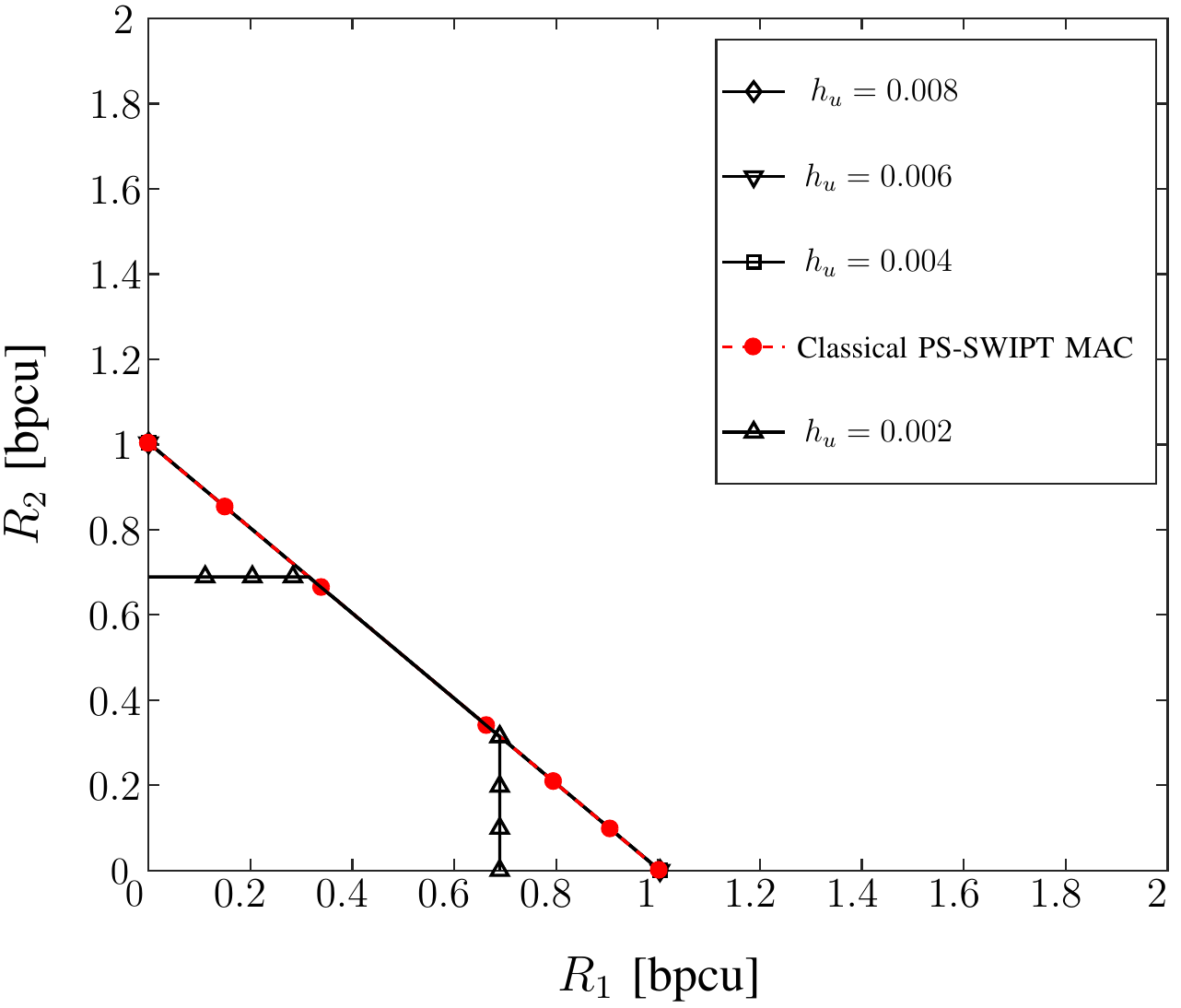}\label{mdrb_sdcoop_h21}}
  \caption{Maximum departure region comparison between Classical PS-SWIPT and PS-SWIPT MAC with cooperation for different values of user-user channel gain and decoding cost function  $\varphi(R)=\beta(2^{2R}-1)$.}
  \label{mdrb_sdcoop} 
\end{figure}

\section{CONCLUSIONS}
We studied two scenarios in the PS-SWIPT MAC wherein the PS scheme was employed at the destination while all the receiving nodes were subjected to decoding cost under non-linear EH model constraints. In the classical PS-SWIPT MAC scenario, we considered two different decoding schemes, i.e., simultaneous decoding and SIC. In each scheme, we initially calculated the maximum departure region; then, the optimal PS factors and the maximum sum rates were derived. The numerical results indicated that by changing the forms of decoding cost function, the performance of both SIC scheme and simultaneous decoding significantly vary especially in the case of constant decoding cost function. Moreover, we derived an analytical optimal solution to jointly optimizing PS factors and power allocation in order to maximize the sum rate in PS-SWIPT MAC with user cooperation. Numerically, the results showed that it is not always efficient to employ cooperation in the PS-SWIPT MAC which is a significant impact of decoding cost and rectifier saturation on the performance of the system.

\section*{APPENDIX A}
\textit{Sketch of proof for Theorem \ref{SIC_sumrate_Theorem}} : Firstly, the solution to $\mathbf{P} \mathbf{2}$ falls on the MDRB characterized by (\ref{simulregion1}) and (\ref{simulregion3}) or on it's Convex Hull Boundary (CHB). To further shrink the feasible set, we assume a pair of points $p^*=(R_1^*, R_2^*)$ and $\hat{p}=(\hat{R}_1, \hat{R}_2)$ on the MDRB characterized by (\ref{simulregion1}) and/or (\ref{simulregion3}). Then for any $0\leq\lambda\leq1$, the sum-rate of  any point on the line segment between $p^*$ and $\hat{p}$ is expressed as
\begin{equation}
\begin{aligned}
\hspace{0cm} 
\|\lambda p^*+(1-\lambda)\hat{p}\|_1&=\lambda R_1^*+ (1-\lambda)\hat{R}_1+\lambda R_2^*+ (1-\lambda)\hat{R}_2\\
&=\lambda(R_1^*+R_2^*)+(1-\lambda)(\hat{R}_1+\hat{R}_2)
\\&\leq max\big(R_1^*+R_2^*,\hat{R}_1+\hat{R}_2\big),
\end{aligned}
\end{equation}
where $\|.\|_1$ stands for $L_1$ norm, indicating that a convex hull does not increase the maximum sum-rate. As a result, regardless of the system's parameters and the form of the functions $\varphi(.)$ and $\psi(.)$, the optimal point in $\mathbf{P} \mathbf{2}$  falls on the MDRB expressed by  (\ref{simulregion1}) or (\ref{simulregion3}). Thus,  confining the feasibility set of $\mathbf{P} \mathbf{2}$ to the MDRB, we reformulate $\mathbf{P} \mathbf{2}$ as 
\begin{align*}
\hspace{-1.7cm}\mathbf{P} \mathbf{2-1}:\max _{{\rho}}~~ & R_1^*+R_2^* \\
\hspace{-0.5cm}\text { {s.t.} }\hspace{0cm}&[\rho,R_1^*,R_2^*]\subseteq(\ref{simulregion1}) \cup (\ref{simulregion3}).
\end{align*}
Since the feasible set of $\mathbf{P} \mathbf{2-1}$ consists of the union of the sets satisfying (\ref{simulregion1}) or (\ref{simulregion3}), we can separately maximize $R_1^*+R_2^*$ with the feasible sets (\ref{simulregion1}) and (\ref{simulregion3}), and then to find the maximum of their solutions. So, we rewrite  $\mathbf{P} \mathbf{2-1}$ as
\begin{align*}
\hspace{-1.7cm}\mathbf{P} \mathbf{2-2}:\max _{{\rho,\hat{\rho}}}\hspace{0cm} & ~~~\text{MAX}(R_1^*+R_2^*,\hat{R}_1+\hat{R}_2) \\
\hspace{-0.5cm}\text { {s.t.} }\hspace{0cm}&[\rho,R_1^*,R_2^*] \subseteq ~(\ref{simulregion1}), \\
\hspace{-0.5cm}\hspace{0cm}&[\hat{\rho},\hat{R}_1,\hat{R}_2] \subseteq ~(\ref{simulregion3}), 
\end{align*}where MAX(.) stands for the maximum value function. Then using the equations (\ref{simulregion1}) and (\ref{simulregion3}), $\mathbf{P} \mathbf{2-2}$ is rewritten as 
\begin{align*}
\hspace{-1.7cm}\mathbf{P} \mathbf{2-3}:\max _{{\rho,\hat{\rho}}}\hspace{0cm} & ~~~\text{MAX}(f_1(\rho),f_2(\hat{\rho})) \\
\hspace{-0.5cm}\text { {s.t.} }\hspace{0cm}&\rho_1^*\leq \rho\leq \rho_c^*, \\
\hspace{-0.5cm}\hspace{0cm}&\rho_2^*\leq \hat{\rho}\leq \rho_c^*, 
\end{align*}wherein the first and the second argument of MAX(.) are independent and non-convex. As a result, solving the derivative equations $f'_1(x)=0$ and $f'_2(x)=0$, Theorem  \ref{SIC_sumrate_Theorem} follows.

\section*{APPENDIX B}

\textit{Sketch of proof for Lemma \ref{SIC_linear_optimal}}: From Lemma \ref{SIC_sumrate_lemma}, (\ref{successphi}) is satisfied with equality. Now, we can add  $R_1$ and $R_2$ by increasing $\rho$ with a small enough value with no constraints violation. So, at least one of (\ref{successr1}) or (\ref{successr2}) holds with equality. First, suppose (\ref{successr2}) is satisfied with equality and (\ref{successr1}) holds with strict inequality. Then, we have 
\begin{equation}  \label{eq15}
R_1=\frac{1}{2}\log\Big(1+\dfrac{\eta\rho a}{\beta}- \dfrac{
	(1-\rho)|h_2|^2P_2}{(1-\rho)N+N_p}\Big).
\end{equation}
So, we have three conditions. First, we show that (\ref{eq15}) satisfies (\ref{successr1}), by a substituting 
and removing the logarithms as
\begin{equation}\label{eq16}
\dfrac{\eta\rho a}{\beta}-\dfrac{(1-\rho)|h_2|^2P_2}{N_p}\leq \dfrac{(1-\rho)|h_1|^2P_1}{(1-\rho)|h_2|^2P_2+N_p},
\end{equation}
 wherein $N_p\gg N$. Simplifying (\ref{eq16}), we derive an   equation as
\begin{equation}\label{eq17}
\begin{aligned}
&-B(\eta a N_p+\beta B)\rho^2 + \big(\eta a N_p(B+N_p)+\beta(2B^2
\\&+CN_p)\big)\rho - \beta (B^2+AN_p+BN_p)\leq 0,
\end{aligned}
\end{equation}
wherein $A$, $B$, and $C$ are defined in Lemma \ref{SIC_linear_optimal}. The left term of (\ref{eq17}) is a quadratic function of $\rho$. 
Since both addition and multiplication of the roots in the quadratic function are positive, the 
both roots are positive. Denoting the smaller root by $\rho_1$, the feasible  $\rho$ in (\ref{eq17}) falls within the interval $0 < \rho \leq \rho_1 < 1$, where
$$\rho_1=0.5\left(1+ \dfrac{\beta B^2+\beta C N_p+\eta a N_p^2-\Delta}{\beta B^2+\eta a N_p B}\right),$$
wherein  $\Delta$ is a real positive value defined in Lemma \ref{SIC_linear_optimal}.
The second condition is to have $\varphi(R_2)\leq P_{EH}$ by satisfying
$$ \dfrac{\beta(1-\rho)|h_2|^2P_2}{N_p}\leq \eta\rho a,$$
which gives the second feasibility interval $\rho \in [\rho_2, 1]$, where 
$$\rho_2=\dfrac{\beta B}{\beta B + \eta a N_p}.$$
As a result, noting that $\rho_2<\rho_1$, the feasible PS factor is bounded as $0 < \rho_2 \leq\rho \leq \rho_1 <1$.

Finally, we prove that the sum rate is increasing in $\rho$ in the feasible interval. From (\ref{eq15}), we have 
\begin{align}\label{eq19}
\begin{aligned}
R_1+R_2=&\frac{1}{2}\log\Big(1+\dfrac{\eta\rho a}{\beta}- \dfrac{
	(1-\rho)|h_2|^2P_2}{N_p}\Big) \\ &+\frac{1}{2}\log\Big(1+\dfrac{(1-\rho)|h_2|^2P_2}{N_p}\Big)
	\\&\hspace{-0.4cm}=\frac{1}{2}\log\Big(1+\dfrac{\eta\rho a}{\beta}+\dfrac{\eta a|h_2|^2P_2(1-\rho)}{\beta N_p}\\
&-\dfrac{(1-\rho)^2|h_2|^4 P_2^2 }{N_p^2} \Big).
\end{aligned}
\end{align}
Note that (\ref{eq19}) is increasing in  $\rho$ iff the argument of the logarithm function is increasing. Thus from (\ref{eq19}), $\dfrac{d\left(R_1+R_2\right)}{d\rho}\geq0$ is satisfied in the interval $0< \rho\leq \rho_c <1$, where
$$ \rho_c=0.5\left(1+\dfrac{\beta B^2+\eta a N_p^2 }{ (\beta B + \eta a N_p) B}\right).$$
Now, to prove that the sum rate is increasing in the interval $[\rho_2, \rho_1]$, we show  
 $\rho_1\leq\rho_c$ as
\begin{equation}\label{eq21}
\dfrac{\beta(B^2+CN_p)+\eta a N_p^2-\Delta}{B(\beta B+\eta a N_p)}\leq\dfrac{\beta B^2 +\eta a N_p^2 }{B(\beta B+\eta a N_p)},
\end{equation}
which is derived from  $\beta C N_p<\Delta$. 
 Hence, setting $\rho=$ $\rho_1$ is feasible and results in the maximum sum rate. In other words, (\ref{successr1}) is also satisfied with equality.
\bibliography{main}
\bibliographystyle{ieeetr}

\end{document}